%% file: main_arxiv.tex
\ifdefined\pdfsuppressptexinfo
\fi
\documentclass[sigconf,nonacm,balance=false]{acmart}

\usepackage{amsmath}
\usepackage{array}
\usepackage{booktabs}

\newtheorem{theorem}{Theorem}
\newtheorem{lemma}[theorem]{Lemma}
\newtheorem{proposition}[theorem]{Proposition}
\newtheorem{corollary}[theorem]{Corollary}

\newcommand{\M}{\mathsf M}
\newcommand{\MP}{\mathsf M_{\mathcal P}}
\newcommand{\calP}{\mathcal P}
\newcommand{\Ind}{\mathbb I}

\title{Reachability-Certified Subteam Decomposition for Locally Interacting Multi-Agent MDPs}

\author{Xiangwu Wang}
\affiliation{%
  \institution{University of Hong Kong}
  \city{}\country{}
}

\author{Chengwei Cao}
\affiliation{%
  \institution{University of California, San Diego}
  \city{}\country{}
}

\author{Hongyuan Tang}
\affiliation{%
  \institution{Carnegie Mellon University}
  \city{}\country{}
}

\begin{document}
\hypersetup{pdfauthor={Xiangwu Wang, Chengwei Cao, and Hongyuan Tang}}

\begin{abstract}
Persistent communication limits force a multi-agent system to decide which agents
may coordinate throughout a rollout.  Current proximity alone is insufficient:
separated agents may interact later, whereas a large pair reward may remain
unreachable until it is heavily discounted.  We introduce Reachability-Certified
Subteam Decomposition (RCSD) for finite multi-agent Markov decision processes with
factorized physical dynamics, finite-range ordered pair rewards, and almost-sure
motion bounds.  RCSD combines a speed-limit lower bound on pairwise contact time
with a reward envelope to form a current-state affinity.  For any capacity-valid
persistent partition, the sum of cut affinities bounds the reward-deletion error of
every unchanged stationary Markov state-feedback policy.  A product of team-optimal policies
for the resulting cut MDP incurs at most twice this certificate in regret against
the centralized optimum.  Both bounds are worst-case tight.  On a controlled
five-agent family, RCSD-Exact reduces aggregate normalized execution regret by
56.0\%, 28.8\%, and 25.3\% relative to uniform, distance-only, and envelope-only
partitions.  A separate stochastic two-dimensional study finds no bound violation
over 384 exact-partition and 1,440 restricted-controller evaluations.  Exact
four-agent evidence favors RCSD over uniform and distance-only grouping; raw
evidence for current contact is borderline and envelope-only is unresolved.
Across balanced 8--20-agent strata, controller-library utility is mixed:
pointwise paired intervals favor RCSD over distance and current contact, include
zero for uniform, and favor envelope-only and Value-MIP over RCSD.  Partition
construction remains subsecond in median up to 100 agents;
this last result does not include affinity formation or MDP planning.
\end{abstract}



\maketitle

\section{Introduction}

Multi-agent planning often exploits sparse transition and reward structure to avoid
a single monolithic controller~\cite{guestrin2001factored,scharpff2016transition}.
Cooperative learning methods likewise factor values or learn communication
protocols~\cite{foerster2016dial,sukhbaatar2016commnet,rashid2018qmix}.  These
approaches do not, by themselves, answer a deployment decision that precedes
planning: under a hard limit on persistent team size, which agents should retain a
shared state channel?

Neither current proximity nor a static reward graph resolves this decision.  Agents
that are separated now may enter interaction range later, and connected components
of local contacts can grow through moving chains.  Conversely, an interaction that
cannot occur for many steps contributes only a discounted tail to present value.
LIMDPs formalize dynamic spatial dependencies and state-dependent communication
groups~\cite{deweese2024limdp}.  Here the resource model is different: a dispatcher
observes the initial state once, forms teams of at most $L$ agents, and maintains
full state sharing within each team but no cross-team channel for the entire
rollout.  The resulting persistent overlay fixes the information scope of every
team planner.

This setting raises a concrete question: \emph{before solving the task, can the
value at risk from a capacity-constrained partition be bounded from the current
state?}  When motion is bounded and pair rewards have finite range, initial
separation yields a policy-independent lower bound on the first possible contact.
Discounting a valid pair-reward envelope from that time gives an affinity whose cut
sum bounds the consequence of ignoring cross-team rewards during planning.
RCSD-Exact minimizes this sum to choose persistent subteams
(Figure~\ref{fig:method}).

\begin{figure*}[t]
  \centering
  \includegraphics[width=\textwidth]{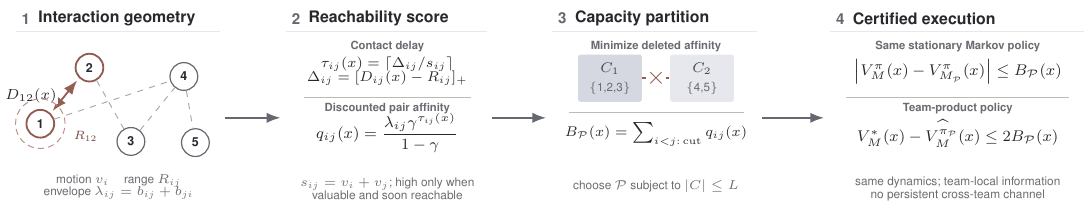}
  \caption{\textbf{Reachability-Certified Subteam Decomposition.}  Current
  geometry, almost-sure speed limits, and the aggregate directed reward envelope
  $\lambda_{ij}=b_{ij}+b_{ji}$ give a contact-time lower bound and discounted
  affinity.  RCSD chooses a capacity-valid persistent partition, deletes cross-team
  rewards only in the planning surrogate, and executes the resulting team policies
  in the original MDP.  The fixed-policy inequality compares the same Markov policy
  in both MDPs; the regret inequality applies to a team-product cut-MDP optimum.}
  \Description{A left-to-right diagram showing spatial agents, pairwise contact-time
  certificates, a capacity-constrained team partition, separate team planners, and
  certified fixed-policy and optimal-policy value bounds.}
  \label{fig:method}
\end{figure*}

RCSD supplies both a pre-planning certificate and an executable decomposition.
It turns LIMDP-style dependence-time reasoning into heterogeneous,
initial-state edge costs derived from distances, speed bounds, interaction
ranges, and both directed reward envelopes.  For any capacity-valid partition,
the cut sum $B_{\calP}(x)$ bounds the value change of an unchanged stationary
Markov policy after cross-team reward deletion.  Under product dynamics and
within-team rewards, the cut MDP separates into team problems, and executing a
team-product cut-MDP optimum in the original MDP incurs at most
$2B_{\calP}(x)$ regret against the centralized optimum.  Both constants are
worst-case tight for signed rewards.

Minimizing the certificate yields RCSD-Exact, while scalable feasible solvers
return partitions with certificates of their own.  We verify the theory on
finite MDPs, compare partition selectors under exact team planning on a
controlled five-agent problem, and evaluate partition construction up to 100
agents.  An omitted-feature control tests whether the advantage persists when the
affinity lacks the feature that determines useful coordination.  Random
two-dimensional obstacle grids then separate exact unrestricted four-agent regret
from controller-library regret at 8--20 agents, testing richer dynamics without
conflating the two estimands.

\section{Related Work}

Factored and transition-independent decentralized MDPs provide a natural context
for local multi-agent planning: they represent local dynamics and structured joint
rewards~\cite{guestrin2001factored,becker2004transition}, while
transition-independent multi-agent MDP solvers
exploit sparse reward interactions without discarding them
\cite{scharpff2016transition}.  Payoff propagation similarly uses sparse
coordination graphs~\cite{kok2006payoff}.  Localized networked-control analyses
show discount-decaying influence with graph distance~\cite{qu2020localized}.
Most closely, LIMDPs combine factorized dynamics, finite-range ordered pair rewards,
bounded movement, and dynamically changing communication groups
\cite{deweese2024limdp}.  Their Dependence-Time Lemma turns bounded motion and a
visibility--reward-range gap into a uniform zero-reward prefix, and their Cutoff
analysis converts that prefix into a discounted value bound.  Unlike LIMDP's
visibility-dependent execution groups and monotone-refining auxiliary cutoff
partition, RCSD selects an arbitrary capacity-valid overlay from heterogeneous
edge certificates and holds it fixed throughout execution.  Its contribution is
not a new coalition objective,
but the combination of an initial-state heterogeneous per-edge certificate, a
fixed-policy deletion bound for every capacity-valid persistent partition, and a
$2B_{\calP}$ execution-regret bound for the resulting product-team policy.

Learned coordination offers a complementary route.  Value
factorization~\cite{sunehag2018vdn,rashid2018qmix}, centralized-training
actor--critic methods~\cite{lowe2017maddpg}, and differentiable communication
\cite{foerster2016dial,sukhbaatar2016commnet} address cooperative learning.
Targeted communication learns whom to message online~\cite{das2019tarmac}, but
does not impose a rollout-persistent block-size cap or return a pre-planning
worst-case value certificate.  VAST learns variable
subteams for value factorization~\cite{phan2021vast}; SOG self-organizes groups
\cite{shao2022sog}; QSCAN represents subteam coordination within value factorization
\cite{huang2022qscan}; and GoMARL and HYGMA learn dynamic grouping structures
\cite{zang2023gomarl,liu2025hygma}.  Correlated Policy Optimization dynamically
allocates a DAG edge budget using dependency scores and analyzes policy optimization
under decomposability conditions~\cite{chen2026cpo}; RCSD instead enforces block
cardinality and certifies every feasible fixed partition before policy optimization.
STAF uses spatial graph cuts for multi-robot formations~\cite{deng2025staf}.  In
contrast to these learned or dynamically structured coordination methods, RCSD supplies
an a priori, planner-independent sufficient loss certificate under a stated
physical model.

From an optimization perspective, capacity-bounded nonnegative edge-weight
coalition formation is an established
optimization problem with hardness and approximation results
\cite{levinger2024bounded}; broader coalition-structure generation also has exact
and anytime algorithms~\cite{rahwan2009anytime}.  Our induced partition objective
is an instance of this problem.  Most directly, Fiscko et al. cluster transition-
independent MDPs around shared controls and clustered value iteration, and
separately optimize reachable state-space size in factored MDPs under a cluster
budget~\cite{fiscko2025clustered,fiscko2023reachability}.  RCSD instantiates this
established objective with heterogeneous initial-state speed-limit delay/envelope
costs and uses their cut sum to bound discounted distortion from deleting
finite-range cross-team rewards.
Standard simulation arguments relate model
perturbations to policy value and transfer optimal policies with a two-sided loss
\cite{kearns2002simulation}.  The distinction here is the reachability-derived
per-edge aggregation used as the persistent-overlay capacity objective.

\section{Model and Capacity-Limited Overlay}
\label{sec:model}

Let $\mathcal I=[n]$ index agents and
$\M=(\mathcal X,\mathcal A,K,r,\gamma)$ be a finite discounted MDP, with
$\mathcal X=\prod_i\mathcal X_i$, $\mathcal A=\prod_i\mathcal A_i$, and
$0<\gamma<1$.  Physical transitions are agent-factorized,
\begin{equation}
K(x'\mid x,a)=\prod_{i=1}^n K_i(x_i'\mid x_i,a_i).
\label{eq:product-kernel}
\end{equation}
Each local state has a position $p_i(x_i)$ in a common metric space
$(\mathcal Z,d)$.  Every positive-probability local transition has displacement at
most $v_i$ under $d$; this is an almost-sure, not expected, motion bound.

Rewards are evaluated from the pre-transition state and action and decompose into
local and ordered pair terms,
\begin{equation}
r(x,a)=\sum_i r_i(x_i,a_i)+
\sum_{i\ne j}r_{ij}(x_i,x_j,a_i,a_j).
\label{eq:reward-factorization}
\end{equation}
For symmetric interaction range $R_{ij}=R_{ji}$, $r_{ij}=0$ whenever
$D_{ij}(x):=d(p_i(x_i),p_j(x_j))>R_{ij}$, and $|r_{ij}|\le b_{ij}$.
The speed, range, and reward-envelope quantities are known valid upper bounds;
underestimation would void the certificate.

At initial state $x$, a dispatcher announces a partition $\calP$ once.  Every
block $C\in\calP$ obeys $|C|\le L$, and the partition remains fixed.  Members of
$C$ persistently share $x_C$ and select
$a_C$ jointly; no cross-team information channel exists.  The cut MDP $\MP$ keeps
the physical kernel and deletes precisely the pair rewards crossing teams:
\begin{equation}
r_{\calP}(x,a)=\sum_i r_i(x_i,a_i)+
\sum_{i\ne j:\,\calP(i)=\calP(j)}r_{ij}(x_i,x_j,a_i,a_j).
\label{eq:cut-reward}
\end{equation}
The cut reward is a planning surrogate.  The resulting team policies are always evaluated
in the original reward model $\M$.

\section{Reachability Certificate}
\label{sec:certificate}

Write $s_{ij}=v_i+v_j$ and $[z]_+=\max\{z,0\}$.  The speed-limit contact-time
lower bound is
\begin{equation}
\tau_{ij}(x)=
\begin{cases}
\left\lceil [D_{ij}(x)-R_{ij}]_+/s_{ij}\right\rceil,&s_{ij}>0,\\
0,&s_{ij}=0,~D_{ij}(x)\le R_{ij},\\
\infty,&s_{ij}=0,~D_{ij}(x)>R_{ij}.
\end{cases}
\label{eq:tau}
\end{equation}
Obstacles, clipping, or the policy can delay contact, so $\tau_{ij}$ need not be
the exact achievable time.  It is sufficient that contact is impossible earlier.

Aggregate both directed envelopes as $\lambda_{ij}=b_{ij}+b_{ji}$ and define, with
$\gamma^\infty=0$,
\begin{equation}
q_{ij}(x)=\frac{\lambda_{ij}\gamma^{\tau_{ij}(x)}}{1-\gamma},\qquad
B_{\calP}(x)=\sum_{\substack{i<j:\\ \calP(i)\ne\calP(j)}}q_{ij}(x).
\label{eq:certificate}
\end{equation}
The affinity $q_{ij}$ charges the full pair envelope at every time from the
certified lower-bound index onward.  It is neither a contact probability nor an
action value.

\begin{lemma}[Speed-limit contact]
\label{lem:contact}
Along every feasible trajectory starting at $x$, both ordered rewards of pair
$(i,j)$ vanish at every pre-transition return index $t<\tau_{ij}(x)$.
\end{lemma}

\begin{proof}
Repeated reverse triangle inequality and the almost-sure step bounds give
$D_{ij}(X_t)\ge D_{ij}(x)-t(v_i+v_j)$.  For every integer
$t<\lceil[D_{ij}(x)-R_{ij}]_+/(v_i+v_j)\rceil$, the right-hand side is strictly
larger than $R_{ij}$ whenever the pair starts outside range.  The zero-speed cases
follow directly.  Finite-range support then makes both ordered rewards zero.
\end{proof}

\begin{theorem}[Policy-wise deletion certificate]
\label{thm:policy}
Let $\pi(a\mid x)$ be any stationary randomized Markov state-feedback policy used
unchanged in $\M$ and $\MP$.  For every initial state $x$,
\begin{equation}
\left|V_{\M}^{\pi}(x)-V_{\MP}^{\pi}(x)\right|\le B_{\calP}(x).
\label{eq:policy-bound}
\end{equation}
\end{theorem}

Here $\pi$ may depend on the full joint state: ``unchanged'' means that the same
action kernel $\pi(\cdot\mid x)$, rather than a policy reoptimized after reward
deletion, is evaluated in both MDPs.

\begin{proof}
The two MDPs have identical policy and transition kernels, hence the same
state--action law.  Their only return difference is the deleted reward.  Lemma
\ref{lem:contact} makes each cut pair's absolute contribution zero before
$\tau_{ij}$ and the two directed envelopes bound it by $\lambda_{ij}$ thereafter.
Triangle inequality and the discounted tail
$\sum_{t=\tau_{ij}}^\infty\gamma^t$ give Eq.~\eqref{eq:policy-bound}.
\end{proof}

The unchanged-policy condition is substantive: a controller that observes deleted
rewards and reacts through reward history need not induce the same future actions.
The appendix gives an executable counterexample.  Reward-blind history-dependent
policies held fixed across models also admit the coupling, but we state the Markov
class used by the optimization result.

\begin{proposition}[Cut-MDP factorization]
\label{prop:factorization}
Under Eqs.~\eqref{eq:product-kernel}--\eqref{eq:cut-reward}, the centralized
optimal value of $\MP$ is attained by a product of team-local stationary Markov
policies $\widehat\pi_{\calP}(a\mid x)=\prod_{C\in\calP}
\widehat\pi_C(a_C\mid x_C)$.
\end{proposition}

\begin{proof}
Collect local and within-team rewards into $r_C$ and write
$K_C=\prod_{i\in C}K_i$.  Then $r_{\calP}=\sum_Cr_C$ and $K=\prod_CK_C$.
The Bellman maximum separates over the Cartesian action blocks $a_C$, so the sum
of team-optimal value functions is the unique cut-MDP optimal value.  Team argmax
policies form a product optimum; discounted finite-MDP optimality is standard
\cite{puterman1994mdp}.
\end{proof}

\begin{corollary}[Execution regret]
\label{cor:regret}
Execute a product-optimal cut policy $\widehat\pi_{\calP}$ in the original MDP.
Then
\begin{equation}
V_{\M}^{*}(x)-V_{\M}^{\widehat\pi_{\calP}}(x)
\le 2B_{\calP}(x).
\label{eq:regret-bound}
\end{equation}
The constants in Eqs.~\eqref{eq:policy-bound} and \eqref{eq:regret-bound} are
worst-case tight when signed rewards are allowed; the $2B_{\calP}$ witness is
worst-case over the cut-MDP-optimal policy returned when cut rewards tie.
\end{corollary}

\begin{proof}
Insert the values of a full-MDP optimum $\pi^*$ and $\widehat\pi_{\calP}$ in the
cut MDP.  The middle difference is nonpositive by cut optimality; Theorem
\ref{thm:policy} bounds each of the two outer differences by $B_{\calP}(x)$.
This is the standard two-sided reward-perturbation transfer
\cite{kearns2002simulation}.  Appendix~\ref{app:proofs} gives the expanded proof
and one-state witnesses attaining $B$ and $2B$.
\end{proof}

The constants are tight.  With two singleton teams in a one-state MDP, a unique
joint action and a constant
deleted reward $b$ give value difference $b/(1-\gamma)=B_{\calP}$.  For the factor
two, let the cut reward be zero for two product actions while their deleted rewards
are $+b$ and $-b$.  Both actions are cut-optimal; if the team solver selects the
negative-reward tie and the original optimum selects the positive one, execution
regret is $2b/(1-\gamma)=2B_{\calP}$.  Thus neither constant can be reduced for
signed rewards without an additional structural or tie-breaking condition.

\begin{corollary}[Approximate team planning]
\label{cor:approximate}
Let $\widetilde\pi_{\calP}(a\mid x)=\prod_{C\in\calP}
\widetilde\pi_C(a_C\mid x_C)$ be a stationary Markov team-product policy used
unchanged in $\M$ and $\MP$, whose aggregate cut-MDP planning error at $x$ is at
most $\bar\epsilon_{\calP}(x)$:
\[
V_{\MP}^{*}(x)-V_{\MP}^{\widetilde\pi_{\calP}}(x)
\le \bar\epsilon_{\calP}(x).
\]
Then its execution regret in the original MDP satisfies
\[
V_{\M}^{*}(x)-V_{\M}^{\widetilde\pi_{\calP}}(x)
\le 2B_{\calP}(x)+\bar\epsilon_{\calP}(x).
\]
\end{corollary}

\begin{proof}
Let $\pi^*$ be optimal in $\M$.  Then
$V_{\M}^*-V_{\MP}^*\le
V_{\M}^{\pi^*}-V_{\MP}^{\pi^*}\le B_{\calP}$ by
Theorem~\ref{thm:policy}.  Insert $V_{\MP}^{*}$ and
$V_{\MP}^{\widetilde\pi_{\calP}}$ between the original-MDP values; the middle
term is at most $\bar\epsilon_{\calP}$ by assumption, and the final reward-model
change is at most $B_{\calP}$ by the same theorem.
\end{proof}

\subsection{Certificate Properties}

The construction separates three quantities that are easily conflated in a static
interaction graph.  The envelope $\lambda_{ij}$ bounds how much the pair can
contribute once active; $\tau_{ij}(x)$ certifies how long it must remain inactive;
and $q_{ij}(x)$ combines them at the current discount.  For fixed envelope and
discount,
\begin{equation}
q_{ij}(\tau+1)=\gamma q_{ij}(\tau),
\label{eq:one-step-delay}
\end{equation}
so one additional certified zero-reward step reduces that pair's charge by exactly
$\gamma$.  The affinity is linear in the valid reward envelope, nonincreasing in
initial separation through the integer-valued contact delay, and nondecreasing when
a larger speed or interaction range permits earlier contact.  These are algebraic
properties of the sufficient bound, not fitted empirical relationships.

Additivity has two consequences.  First, Theorem~\ref{thm:policy} holds for every
capacity-valid partition, not only the partition selected by RCSD.  The solver can
therefore return an explicit certificate even if stopped early or replaced by a
different feasible coalition routine.  Second, exact minimization of
$B_{\calP}(x)$ selects the smallest theorem-derived upper bound within the stated
partition class.  It does \emph{not} necessarily select the partition with minimum
realized regret: deleted rewards can cancel, remain unreachable after the certified
prefix, or be irrelevant to the optimal action.  The hindsight regret oracle in the experiment
measures this distinction but is never available to RCSD.

Several edge cases clarify interpretation.  A zero-speed pair initially outside
range has $\tau_{ij}=\infty$ and zero certified affinity, so separating it cannot
change value through that pair under the model.  A pair already in range has
$\tau_{ij}=0$ and receives the undiscounted infinite-horizon envelope
$\lambda_{ij}/(1-\gamma)$, representing the worst-case repeated tail rather than
a prediction that the reward will recur.  Conservative overestimates of speed,
range, or reward magnitude
preserve validity while weakening selection resolution.  Underestimates can break
the zero-prefix or tail bound and therefore receive no guarantee.

The certificate is initial-state conditional.  Recomputing it later creates a new
optimization problem because a new partition changes the information structure,
cut reward, and policy being evaluated.  The fixed-overlay setting keeps setup
communication distinct from persistent per-decision links.

These results form a hierarchy of guarantees.  The policy-wise certificate uses
finite-range pre-transition rewards,
almost-sure motion bounds, both directed reward envelopes, and one stationary
Markov policy held fixed across the two reward models.  These conditions establish
the common trajectory law and bound every cut-pair tail.  Within our global
product-kernel model, this policy-wise step uses only the common kernel and does
not invoke factorization.  The execution-regret result adds product dynamics and
actions, complete cross-team reward deletion in the surrogate, and team-optimal
planning so that the cut optimum is an executable product policy.  A bounded-error
team planner adds $\bar\epsilon_{\calP}$ through
Corollary~\ref{cor:approximate}.  Both guarantees are conditional on the initial
state, fixed partition, and valid parameter bounds.

\section{Certificate-Guided Partitioning}
\label{sec:method}

RCSD denotes the certificate construction in Eq.~\eqref{eq:certificate} together
with a capacity-valid partition solver.  RCSD-Exact minimizes the certificate
exactly; RCSD-Greedy and RCSD-MnM-sum are scalable heuristics whose outputs retain
valid certificates but need not minimize them.  After partitioning, exact
team-local planning returns the product-optimal cut policy used by
Corollary~\ref{cor:regret}.  Because total pair affinity is partition-independent,
\begin{equation}
\min_{\calP:\,|C|\le L}B_{\calP}(x)
\Longleftrightarrow
\max_{\calP:\,|C|\le L}
\sum_{C\in\calP}\sum_{i<j\in C}q_{ij}(x).
\label{eq:capacity-objective}
\end{equation}
Equation~\eqref{eq:capacity-objective} is the known bounded edge-weight coalition
objective~\cite{levinger2024bounded}.  For exact optimization, let
$W(C)=\sum_{i<j\in C}q_{ij}$ and fix the least-indexed agent in each remaining set
$S$.  The subset recurrence
\begin{equation}
F(S)=
\max_{\substack{C\subseteq S:\,\min S\in C\\1\le |C|\le L}}
\{W(C)+F(S\setminus C)\},\qquad F(\varnothing)=0,
\label{eq:subset-dp-main}
\end{equation}
enumerates the unique block containing $\min S$ and then recurses.  Every feasible
partition appears once at that decision level, so additivity of retained affinity
gives optimal substructure.  RCSD-Exact uses this dynamic program at small $n$.

At larger $n$, RCSD-Greedy repeatedly merges the feasible pair of blocks with the
largest positive rescued affinity.  RCSD-MnM-sum instead applies deterministic
matching and contraction with cross-block affinities summed after each contraction.
Both return capacity-valid partitions and therefore retain the policy-wise
certificate for their outputs, but neither implementation is assigned an
approximation ratio here.  Full pseudocode, tie-breaking, and correctness details
appear in Appendix~\ref{app:algorithms}.

Pair construction takes $O(n^2)$ time and space.  For fixed capacity, our exact
subset DP takes $O(n^L L^2+2^n n^{L-1})$ time and
$O(2^n+n^L+n^2)$ space.  The direct greedy implementation uses at most $n-1$
merges and a conservative $O(n^3L^2)$ time bound; no approximation guarantee is
asserted.  A fully connected directed channel within each team uses
\begin{equation}
\sum_{C\in\calP}|C|(|C|-1)\le n(L-1)
\label{eq:communication}
\end{equation}
persistent links per decision, versus $n(n-1)$ for one centralized team.  This is
a link count, not measured bytes or latency, and excludes one-time membership
announcement.  If each agent has at most $S$ states and $A$ actions, a size-$L$
team has at most $S^L$ joint states and $A^L$ joint actions.  RCSD caps the local
planning dimension but does not remove its exponential dependence on $L$.

The certificate also exposes the capacity trade-off.  Let $\calP_L^*$ attain the
minimum certificate over partitions with block size at
most $L$, and write $B_L^*(x)=B_{\calP_L^*}(x)$.  Because the feasible sets are nested,
$B_{L+1}^*(x)\le B_L^*(x)$, while Eq.~\eqref{eq:communication} increases the
allowed link envelope linearly in $L$ and the tabular planning dimension can grow
exponentially.  Given a declared value-loss tolerance $\delta$ and a valid bound
$\bar\epsilon_L(x)$ on the aggregate cut-MDP planning error of the unchanged
stationary Markov team policy returned for $\calP_L^*$, a designer may select the
smallest acceptable $L$ satisfying
$2B_L^*(x)+\bar\epsilon_L(x)\le\delta$.  This is a sufficient certification rule,
and it allows $\bar\epsilon_L$ to vary independently of capacity and actual
communication to underfill its envelope.  Conservative physical envelopes
preserve the rule but can require a larger capacity.  For a heuristic partition
$\calP_L$, the same rule uses its returned $B_{\calP_L}(x)$ and corresponding
planning-error bound.

\begin{table}[t]
\caption{RCSD construction for one initial state.}
\label{tab:rcsd-algorithm}
\centering
\small
\renewcommand{\arraystretch}{1.02}
\begin{tabular}{@{}>{\raggedright\arraybackslash}p{0.12\columnwidth}>{\raggedright\arraybackslash}p{0.80\columnwidth}@{}}
\toprule
Stage & Operation\\
\midrule
Input & Initial state $x$, capacity $L$, discount $\gamma$, valid
$(v_i,R_{ij},b_{ij})$, a partition solver, and a team planner.\\
Edges & For every $i<j$, compute $\lambda_{ij}$, $\tau_{ij}(x)$, and
$q_{ij}(x)$ by Eqs.~\eqref{eq:tau}--\eqref{eq:certificate}.\\
Partition & Solve Eq.~\eqref{eq:capacity-objective} exactly or heuristically,
subject to $|C|\le L$.\\
Planning & Form each induced cut-MDP component and obtain
$\widehat\pi_C$ independently.\\
Output & Return $\calP$, $B_{\calP}(x)$, and
$\widehat\pi_{\calP}=\prod_C\widehat\pi_C$.\\
Run & Announce the fixed teams and execute their policies in the
original MDP.\\
\bottomrule
\end{tabular}
\end{table}

The partition is constructed without planned values or realized regret, so its
certificate is available before team planning.  Theorem~\ref{thm:policy} applies
to every feasible output; Corollaries~\ref{cor:regret} and
\ref{cor:approximate} connect the returned team policy to centralized execution
regret.

There are two distinct computational bottlenecks.  Pair construction and scalable
partitioning are polynomial in the tested implementation.  Team planning remains
exponential in the capacity for dense tabular models, reflecting the underlying
joint-decision problem rather than the partitioning routine.  Accordingly, our large-
$n$ experiment isolates partition construction, while exact five-agent experiments
measure policy quality.  We report the two arms separately because fast graph
partitioning does not imply fast end-to-end multi-agent planning.

Persistence is essential to the information constraint.  If membership followed
the current visibility graph, a chain of local contacts
could merge many agents into one decision group, and membership could change within
the rollout~\cite{deweese2024limdp}.  The imposed capacity would then no longer
bound the information scope of a local planner.  A persistent partition makes the
resource statement exact: each policy sees at most $L$ agents, while physical
agents continue to move and cross-team rewards remain present during execution.
Deleting those rewards only in the surrogate is what connects the communication
decision to Theorem~\ref{thm:policy}; physically removing interactions would define
a different control problem.

\section{Experimental Design}
\label{sec:experiments}

We evaluate RCSD through finite-MDP verification, exact partition utility when the
certificate inputs are informative or incomplete, heuristic objective quality,
partition-construction scaling, and stochastic two-dimensional transfer.  The
omitted-feature family makes compatible teams depend on information absent from
the certificate.  Every policy-quality estimand uses exact optimization over its
stated policy class, avoiding training variance.

We verify the assumptions and bounds at three levels.  We exhaustively test the
contact lemma over 648 bounded-motion settings
(124,380 transitions) and evaluate the policy-wise bound for 2,048 deterministic
policies over signed asymmetric rewards.  We add 200 randomized policies on 40
stochastic instances and 1,716 statewise factorization/regret comparisons across
96 capacity partitions.  Policy values come from linear solves and optima from
converged dynamic programming.  Counterexamples obtained by changing reward
timing, omitting a directed envelope, or permitting reward-history reaction mark
the boundary of the assumptions; separate witnesses attain both constants.

To isolate partition selection while exactly evaluating every feasible
equal-communication partition, we use the Meeting-Port Corridor, a controlled
five-agent family with seven
corridor positions, three unit-step actions, deterministic clipped motion,
$\gamma=0.85$, and interaction range zero.  Thus the centralized table has
$7^5=16{,}807$ states and
$3^5=243$ joint actions.  Agent $i$ receives home-deviation and movement costs,
while pair $(i,j)$ earns $\lambda_{ij}$ only when both agents wait at its meeting
port.  The informative family starts at homes $(0,1,3,5,6)$; each port is a floor/ceiling midpoint,
so first possible port contact equals $\lceil|h_i-h_j|/2\rceil$.  Envelopes are
uniform on $[0.45,1.35]$.  The omitted-feature control starts every agent at 3,
uses near-equal envelopes on $[0.78,0.82]$, and randomly assigns ports from a fixed
multiset.  Every certified delay is then zero, while useful grouping depends on
incompatible ports that RCSD does not observe.  This control changes several
construction parameters and therefore limits family-independent claims; it is not a
single-variable causal ablation.

We evaluate 4,500 instances from each family, each at $L\in\{2,3\}$.  At $L=2$ we
evaluate all 15 partitions of shape $2+2+1$; at
$L=3$, all ten of shape $3+2$.  This gives 18,000 seed/capacity strata and
225,000 complete equal-shape partition rows.  The 4,500 underfilled
visibility-components rows are kept separate.  For every partition, team policies
are solved in the cut MDP
and their deterministic infinite-horizon return is evaluated in the original MDP
by exact prefix--cycle summation.  Regret is normalized by
$\sum_{i<j}\lambda_{ij}/(1-\gamma)$.

For comparison, the uniform baseline is the exact mean over the feasible set.  The
distance-only
selector uses $1/(1+D_{ij})$, and the envelope-only selector removes the delay and
uses $\lambda_{ij}/(1-\gamma)$; RCSD-Exact uses $q_{ij}$.  A current-contact
selector and a hindsight regret oracle are secondary diagnostics.  The paired unit is the seed, with
$L=2,3$ averaged
within seed before each primary test ($n=4{,}500$).  For comparator $Q$, aggregate regret
reduction is $1-\sum_s g_{\mathrm{RCSD},s}/\sum_s g_{Q,s}$.  We report paired mean
differences, paired standardized effects $d_z$~\cite{lakens2013effects}, and win
rates.  Percentile intervals use 10,000 paired-seed
bootstrap draws~\cite{efron1993bootstrap}; 100,000 one-sided paired sign flips are
Holm-corrected over the three primary comparisons~\cite{holm1979multiple}.
Certificate--regret ranking uses within-stratum Spearman correlation
\cite{spearman1904association} with a seed-clustered bootstrap.
Distance-only and envelope-only have the same partition support, planner, and link
budget as RCSD, so they isolate the two factors in its affinity.  They are not
intended as claims of dominance over learned or dynamically regrouping methods
with different information structures.

We compare three solver variants: RCSD-Exact, RCSD-Greedy, and RCSD-MnM-sum.
Each uses 240 seeds at
$n=8$, $L\in\{2,4\}$, and at $n=12$, $L\in\{2,3,4\}$.  Scaling uses 300 seeds for each
combination $n\in\{24,48,100\}$ and $L\in\{2,4\}$ with random 2-D positions,
speeds, ranges, and envelopes.  This arm constructs partitions only; it does not
solve a 100-agent MDP or support a 100-agent value claim.

Appendix~\ref{app:twod} adds connected obstacle grids with stochastic slip and
heterogeneous speeds, ranges, and rewards.  In 128 four-agent instances, exact
centralized and team dynamic programming measure unrestricted stationary-Markov
regret over all three $2+2$ partitions.  In 120 maps with
$n\in\{8,12,16,20\}$, a partition-independent library of four waypoint
controllers makes end-to-end evaluation tractable; this tier reports only
controller-library regret, with $L=2,4$ averaged within map for paired inference.

\section{Results}
\label{sec:results}

\begin{table*}[t]
\caption{Summary of empirical results.  Utility effects average capacities within each of
4,500 instances; rank correlation is computed within instance/capacity strata.}
\label{tab:evaluation-overview}
\centering
\small
\begin{tabular}{lll}
\toprule
Evaluation & Metric & Result\\
\midrule
Partition utility & ARR vs. Uniform / Distance / Envelope & .5604 / .2880 / .2531\\
Certificate ranking & median $\rho$ / clustered 95\% CI & .7143 / [.7091,.7212]\\
Greedy objective quality & median / p90 / max normalized gap & .01110 / .05287 / .13845\\
100-agent construction & RCSD-Greedy / RCSD-MnM-sum median at $L=4$ & .0545 / .2333 s\\
Omitted-feature control & ARR / signed median $\rho$ / mean-diff. 95\% CI & .03096 / .1643 / [.00155,.00281]\\
Random 2-D transfer & exact utility; large library utility / bound violations & positive; mixed / 0\\
\bottomrule
\end{tabular}
\end{table*}

Across the verification suite, all finite-state instances satisfy the scoped
inequalities, factorization identity,
and capacity constraints, with largest numerical residual
$1.42\times10^{-14}$.  One-state fixtures attain the constants $B$ and $2B$.
Counterexamples with post-transition rewards, one-sided envelopes, and
reward-reactive policies violate the unmodified formula, showing that the
corresponding assumptions are substantive.
Table~\ref{tab:validation-main} summarizes coverage; full constructions and
boundary outcomes are in Appendix~\ref{app:validation}.

\begin{table}[t]
\caption{Finite-MDP verification of the theorem predictions.}
\label{tab:validation-main}
\centering
\small
\begin{tabular}{lrr}
\toprule
Property & Evaluations & Outcome\\
\midrule
Pre-contact reachability & 124,380 & 0 early contacts\\
Policy-wise deletion & 2,248 & 0 violations\\
Product/regret bounds & 1,716 & 0 violations\\
Tightness witnesses & 2 & $B$ / $2B$ attained\\
Numerical solutions & all & max.\ $1.42{\times}10^{-14}$\\
\bottomrule
\end{tabular}
\end{table}

The random two-dimensional study extends this audit to obstacles, stochastic
motion, and heterogeneous physical parameters.  Across 384 exact four-agent
partitions and 1,440 large-tier method evaluations, no deletion or regret bound is
violated.  In the exact tier, RCSD's aggregate unrestricted regret is 21.6\%,
24.1\%, and 14.8\% lower than Uniform, Distance, and Current.  Paired raw evidence
is clear for Uniform and Distance, borderline for Current, while the 2.6\%
Envelope difference is inconclusive.  The larger controller-library result is
mixed.  Pooled equally across the four $n$ strata, raw comparator-minus-RCSD
differences (pointwise 95\% intervals) for Distance,
Current, Uniform, Envelope, and Value-MIP are $10.05$ $[5.36,14.97]$, $4.90$
$[.91,8.81]$, $-1.89$ $[-8.83,4.78]$, $-5.80$ $[-11.23,-.57]$, and $-5.44$
$[-8.66,-2.30]$.  RCSD's median/p95 regret-to-certificate ratios are .076/.206
in the exact tier and .070/.175 in the library tier.  Appendix~\ref{app:twod}
reports normalization details and selector-plus-team-planning costs.

\begin{figure*}[t]
  \centering
  \includegraphics[width=\textwidth]{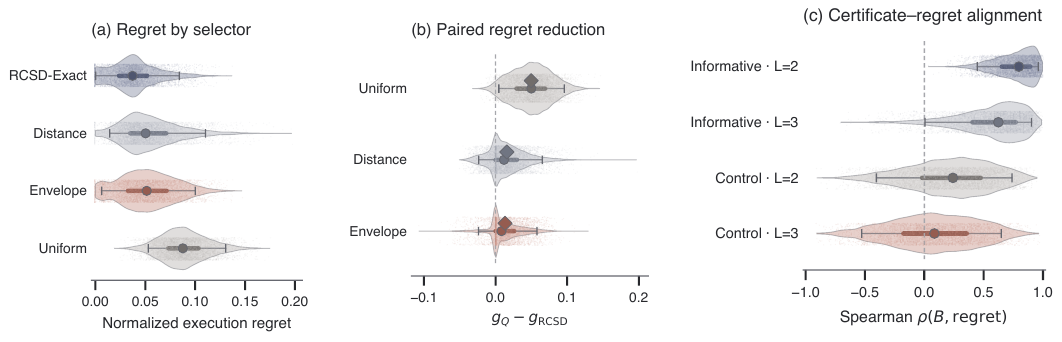}
  \caption{\textbf{Partition quality on Meeting-Port Corridor families.}
  Density envelopes and faint unit marks use all 4,500 paired seeds after
  averaging capacities in (a,b).  Thin/thick summaries in (a,c) are 5--95\%
  ranges/IQRs, circles are medians, and diamonds with capped segments in (b)
  are paired means with 95\% bootstrap intervals for
  $g_Q-g_{\mathrm{RCSD}}$, so positive values favor RCSD.  Panel (c) shows
  4,500 correlations per family--capacity cell (18,000 total).  RCSD-Exact
  lowers regret in the informative family; the control shift is much smaller.}
  \Description{Three distribution panels show normalized regret, paired regret
  reductions with confidence intervals, and certificate-to-regret rank
  correlations for informative and omitted-feature controlled task families.}
  \label{fig:utility}
\end{figure*}

Figure~\ref{fig:utility} and Table~\ref{tab:utility} report all primary
informative-family
comparisons.  Mean normalized regret is 0.03910 for RCSD-Exact, versus 0.08895 for
Uniform, 0.05492 for Distance, and 0.05235 for Envelope.  Across 9,000
instance/capacity strata, median Spearman correlation
between $B_{\calP}$ and actual regret is 0.7143; 8,781/9,000 are positive, and the
seed-clustered 95\% interval for the median is $[0.7091,0.7212]$.  These results
support the mechanism on the controlled informative construction, not a universal
ordering of partitions.

\begin{table}[t]
\caption{Informative-family paired utility results.  CI is for mean $g_Q-g_{\mathrm{RCSD}}$;
the paired unit is one of 4,500 seeds after averaging capacities.}
\label{tab:utility}
\centering
\small
\begin{tabular}{lrrrr}
\toprule
$Q$ & ARR & Mean diff. [95\% CI] & $d_z$ & Holm $p$\\
\midrule
Uniform & .5604 & .04985 [.04905,.05068] & 1.791 & .000030\\
Distance & .2880 & .01581 [.01504,.01656] & .604 & .000030\\
Envelope & .2531 & .01325 [.01255,.01396] & .554 & .000030\\
\bottomrule
\end{tabular}
\end{table}

Seed-paired bootstrap 95\% intervals for ARR are
$[.5530,.5682]$, $[.2763,.2994]$, and $[.2415,.2647]$ against Uniform,
Distance, and Envelope, respectively.

The three comparisons isolate different information.  Uniform measures the value
of selecting any equal-communication partition; Distance keeps geometry but removes
reward magnitude; Envelope keeps reward magnitude but removes contact delay.  Since
all three exact mechanisms optimize over the identical shape support, their regret
differences cannot be attributed to different link counts or a stronger partition
solver.  The current-contact selector is a secondary diagnostic over the same fixed
shape; the visibility-components diagnostic is secondary because its $L=3$
informative-family partition keeps
shape $2+2+1$ rather than the primary $3+2$ shape.  The hindsight regret oracle is
reported only to quantify how much selection information remains outside the
certificate.

The within-stratum ranking in Figure~\ref{fig:utility}(c) is the more informative
measure of selection resolution because every correlation
compares partitions of the same generated MDP and capacity.  Its clustered interval
retains the seed, rather than the individual partition, as the source of independent
instance-level variation.  This rank statistic measures partition-selection
resolution; it neither calibrates $B_{\calP}$ to typical regret nor establishes
numerical tightness of the $2B_{\calP}$ bound.

In the omitted-feature control family, ARR between RCSD-Exact and Uniform is
0.03096.  The paired mean difference is 0.00219 (interval
$[0.00155,0.00281]$), with $d_z=0.1000$ and win rate 0.5164.  Its signed median
rank correlation is 0.1643; the clustered interval is $[0.1515,0.1758]$.
The ARR is about one eighteenth of the informative
Uniform comparison, the standardized effect is small, and the win rate is near
one half.  Overlapping partitions can induce residual rank association even when
port compatibility is absent from $q_{ij}$.  Together, the two families show that
RCSD's advantage is strongest when the certificate inputs encode
coordination-relevant variation.  Because several generator parameters differ,
the control is not a single-variable ablation; additional control diagnostics
appear in Appendix~\ref{app:control}.

\begin{figure*}[t]
  \centering
  \includegraphics[width=\textwidth]{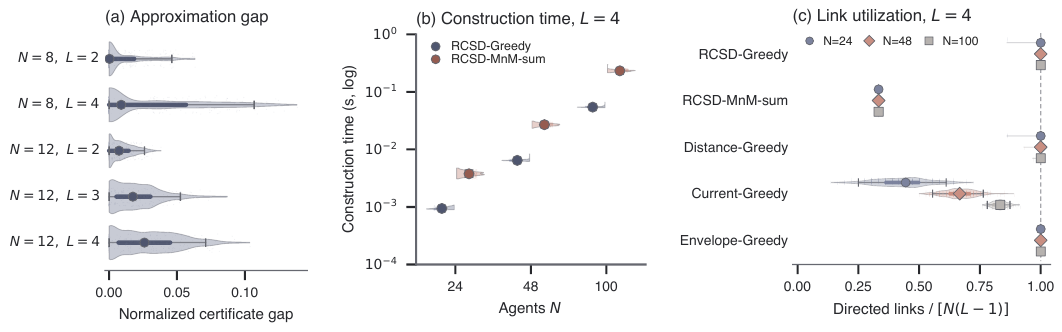}
  \caption{\textbf{Certificate optimization and scaling.}  Density envelopes
  and faint unit marks use all seed-level values (ties overplot); thin/thick
  intervals are empirical 5--95\% ranges/IQRs and circles are medians.
  Panel (a) uses 240 seeds per exact-oracle cell; (b) uses 300 runtime seeds
  per method and agent count at $L=4$ (log scale); and (c) keeps the 300 seeds
  in each method--agent-count stratum separate.  The $x=1$ reference is the
  $N(L-1)$ capacity ceiling.  Except RCSD-MnM-sum, panel (c) reports the
  corresponding greedy construction.}
  \Description{Three seed-level quantile-forest panels show greedy certificate
  gaps, solver-runtime distributions as agent count grows, and within-team
  communication utilization for several mechanisms.}
  \label{fig:scaling}
\end{figure*}

\begin{table*}[t]
\caption{Certificate gap to the exact capacity-partition optimum over 240
instances per cell.}
\label{tab:solver-quality-main}
\centering
\small
\begin{tabular}{rrlrrr}
\toprule
$n$ & $L$ & Solver & Median & p90 & Maximum\\
\midrule
8 & 2 & RCSD-Greedy & .00025 & .04024 & .06324\\
8 & 2 & RCSD-MnM-sum & .00000 & .00000 & .00000\\
8 & 4 & RCSD-Greedy & .00890 & .09264 & .13845\\
8 & 4 & RCSD-MnM-sum & .30076 & .34496 & .39753\\
12 & 2 & RCSD-Greedy & .00725 & .02165 & .03823\\
12 & 2 & RCSD-MnM-sum & .00000 & .00000 & .00000\\
12 & 3 & RCSD-Greedy & .01747 & .04325 & .08698\\
12 & 3 & RCSD-MnM-sum & .12541 & .14642 & .17929\\
12 & 4 & RCSD-Greedy & .02600 & .05885 & .10370\\
12 & 4 & RCSD-MnM-sum & .22599 & .25123 & .26840\\
\bottomrule
\end{tabular}
\end{table*}

Across 1,200 oracle-sized cases, RCSD-Greedy's normalized certificate gap has median
0.01110, p90 0.05287, and maximum 0.13845
(Figure~\ref{fig:scaling} and Table~\ref{tab:solver-quality-main}).  This is
empirical proximity, not an approximation theorem.  Every large-instance output
is capacity valid.  At $n=100,L=4$, median (95th-percentile) seconds are
0.0545 (0.0569) for RCSD-Greedy and 0.2333 (0.2469) for RCSD-MnM-sum.
Timings begin after affinity matrices have been formed and exclude instance
generation, affinity construction, and team-MDP planning; they establish
partition-construction feasibility only.  These implementation-specific
timings are not hardware-independent performance guarantees.

Across all measured cells, RCSD-Greedy is faster and its median normalized gap
remains below 0.027.  RCSD-MnM-sum is exact in the two tested $L=2$ cells and
less accurate at larger capacities under our summed-contraction completion.
These are empirical comparisons rather than approximation
guarantees~\cite{levinger2024bounded}.

Figure~\ref{fig:scaling}(c) also distinguishes capacity from utilization.
Envelope-Greedy fills the allowance in all plotted cases; RCSD-Greedy and
Distance-Greedy do so in 898/900 and 894/900 cases.  RCSD-MnM-sum at $L=4$
uses one third of the maximum
because its perfect first-round matching leaves no singleton, so the implementation
stops with pairs.  Current's median fraction rises from 0.444 to 0.667 to 0.833
at $N=24,48,100$, with overall range 0.139--0.913.
Underfilling can reduce communication, but it is not a free performance comparison;
that is why only exact equal-shape partitions enter the primary utility tests.

\section{Discussion}

Informative-family selectors share the instance, capacity, team planner,
partition set, and communication budget, so their differences isolate edge
scoring.  The results support $q_{ij}$ when its inputs encode
coordination-relevant variation.  In the omitted-feature control, certified
delays are zero and port compatibility is absent from $q_{ij}$; standardized
effect and win-rate shifts are much smaller.  Because other generator parameters
also differ, this contrast defines a scope boundary rather than a causal ablation.

Capacity has different effects on certificate quality and computation.  For the
inequality-constrained objective $|C|\le L$, the feasible set at capacity
$L$ is contained in the feasible set at $L+1$.  The exact minimum certificate is
therefore nonincreasing in $L$.  At the extremes, $L=1$ cuts every pair, while
$L=n$ permits one team and $B_{\calP}=0$.  The latter recovers centralized
information and joint planning rather than resolving the resource constraint.
Increasing $L$ can reduce the worst-case deletion budget, but it raises the
directed-link envelope $n(L-1)$ and expands a tabular team's state--action space
exponentially.  The capacity rule in Section~\ref{sec:method} makes this trade-off
explicit: a designer can choose the smallest $L$ whose certificate and planning
error meet a declared loss tolerance, then apply a separate application cost
model for communication and computation.

The 100-agent experiment isolates graph partitioning and does not time affinity
construction or team-MDP planning.  The corridor instead permits exact evaluation
of every feasible partition, while the random two-dimensional study bridges these
regimes: unrestricted planning remains exact at four agents, and 8--20-agent
evaluation is exact only within a fixed controller library.  This separation
prevents scalable construction from being mistaken for scalable unrestricted
planning and keeps each selector comparison within the same partition shape and
policy class.

The certificate is intentionally conservative.  Equation~\eqref{eq:certificate}
pessimistically charges the full pair envelope at
every step after the certified contact-time lower bound.  This makes the quantity policy-uniform
and computable before planning.  A sharper score could use obstacles, reachable
sets, occupancy bounds, or policy restrictions, but each requires additional
verified information and a new argument bounding discounted deleted reward.
Empirical contact frequencies or point predictions from a learned model may be
useful ranking features, but they enter the guarantee only through valid
uncertainty bounds.  RCSD can therefore become sharper as additional certified
reachability information becomes available without changing the perturbation
argument.

Exact RCSD minimizes the available upper bound within its partition class; it
need not minimize realized regret.  Deleted rewards may cancel, contact may occur
later than the speed limit permits, and some interactions may be irrelevant to
the optimal action.  Conversely, a heuristic solver returns a valid certificate
for its feasible partition even when it misses the minimum.  This distinction
explains why the experiments report certificate validity, objective quality, and
execution regret separately: each measures a different property of the method.

By treating subteams as persistent, RCSD makes communication membership a resource
decision rather than a
visualization of the current reward graph.  Its scientific role is to connect a
fixed bound on each planner's information scope with a bound on execution loss in
the unchanged physical MDP.  Cross-team interactions remain present during
execution; only the team-planning surrogate omits them.  This connection between
an executable information structure and policy value is what distinguishes the
setting from state-dependent visibility groups.

These guarantees have a specific scope.  The policy-wise certificate assumes
finite state and action spaces,
pre-transition finite-range rewards, valid almost-sure motion and reward
envelopes, symmetric pair ranges, and the same stationary Markov policy in both
reward models.  Its proof uses the common trajectory kernel without invoking the
global product factorization.  The execution-regret corollary additionally uses product transitions and actions,
complete deletion of cross-team rewards in the surrogate, and team-optimal cut
planning so that the selected policy is independently executable.  With an
approximate stationary Markov team planner, Corollary~\ref{cor:approximate} adds
its bounded cut-MDP planning error.  Coupled dynamics, shared action constraints,
higher-order rewards, dynamic repartitioning, and general Dec-POMDP observation
structures require different decompositions or guarantees
\cite{bernstein2002complexity,oliehoek2016decpos}.

\section{Conclusion}

RCSD turns heterogeneous reachability and finite-range reward envelopes into a
computable certificate for capacity-limited persistent subteams.  Its
cut-affinity objective bounds unchanged-policy reward deletion and, under product
dynamics, centralized execution regret.  Exact and stochastic tests produced no
certificate violations; partition experiments scaled to 100-agent affinity
graphs.  Across selectors, mixed large-tier rankings show that certificate
tightness remains central to empirical team selection.  Together, these results
establish RCSD as a certifiable basis for persistent multi-agent decomposition.

\bibliographystyle{ACM-Reference-Format}
\bibliography{references}

\appendix
\input{appendix}

\end{document}

%% file: appendix.tex
\section{Assumptions and Expanded Proofs}
\label{app:proofs}

This appendix makes every scope condition used by the certificate explicit and
expands the compact arguments in the main paper.  The setting is a finite
discounted MDP.  Local states need not consist only of positions: the map
$p_i:\mathcal X_i\rightarrow\mathcal Z$ extracts position in a shared metric
space.  The motion inequality holds almost surely for every transition with
positive probability.  An expectation-only speed bound is insufficient for the
pathwise zero-reward prefix.  Rewards use $(X_t,A_t)$ before the transition to
$X_{t+1}$, pair ranges are symmetric, and both ordered envelopes are included in
$\lambda_{ij}$.  The partition is selected for one initial state and held fixed.

\subsection{Pathwise contact bound}

\begin{proof}[Expanded proof of Lemma~\ref{lem:contact}]
Let $X_t=(X_{1,t},\ldots,X_{n,t})$ and abbreviate
$D_t=d(p_i(X_{i,t}),p_j(X_{j,t}))$.  Repeated reverse triangle inequality gives,
pathwise,
\begin{align}
D_t
&\ge D_0-\sum_{k=0}^{t-1}\bigl[
d(p_i(X_{i,k+1}),p_i(X_{i,k}))\nonumber\\
&\hspace{33mm}+
d(p_j(X_{j,k+1}),p_j(X_{j,k}))\bigr]\nonumber\\
&\ge D_0-t(v_i+v_j).
\label{eq:pathwise-distance}
\end{align}
If $D_0\le R_{ij}$, then $\tau_{ij}=0$ and there is no integer
$t<\tau_{ij}$.  Otherwise write $\delta=D_0-R_{ij}>0$.  When
$s_{ij}=v_i+v_j>0$, every integer $t<\lceil\delta/s_{ij}\rceil$ satisfies
$ts_{ij}<\delta$, and Eq.~\eqref{eq:pathwise-distance} yields
$D_t>D_0-\delta=R_{ij}$.  When $s_{ij}=0$, the same equation yields
$D_t\ge D_0>R_{ij}$ for all $t$ and $\tau_{ij}=\infty$.  Finite-range support
makes $r_{ij}$ and $r_{ji}$ zero in either case.  The index is unshifted precisely
because the return term at index $t$ evaluates the pre-transition state $X_t$.
\end{proof}

The bound can be conservative.  Walls, clipped motion, stochastic slip, or an
agent's policy can postpone or prevent contact.  Such delays do not threaten the
certificate: they only extend the actual zero-reward prefix beyond the certified
one.

\subsection{Policy-wise value difference}

For each cut pair, define its deleted reward and their sum by
\begin{align}
\delta_{ij}(x,a)
&=r_{ij}(x_i,x_j,a_i,a_j)\nonumber\\
&\quad+r_{ji}(x_j,x_i,a_j,a_i),\nonumber\\
\Delta_{\calP}(x,a)
&=r(x,a)-r_{\calP}(x,a)\nonumber\\
&=\sum_{\substack{i<j:\\ \calP(i)\ne\calP(j)}}\delta_{ij}(x,a).
\label{eq:deleted-reward}
\end{align}

\begin{proof}[Expanded proof of Theorem~\ref{thm:policy}]
The models share the initial state, transition kernel, and stationary action
kernel $\pi(a\mid x)$.  Couple their randomness to produce the same
$(X_t,A_t)$ path.  Lemma~\ref{lem:contact} and both directed envelopes give
\begin{equation}
|\Delta_{\calP}(X_t,A_t)|
\le\sum_{\substack{i<j:\\ \calP(i)\ne\calP(j)}}
\lambda_{ij}\Ind\{t\ge\tau_{ij}(x)\}.
\label{eq:deleted-pathwise}
\end{equation}
Consequently,
\begin{align}
\left|V_{\M}^{\pi}(x)-V_{\MP}^{\pi}(x)\right|
&=\left|\mathbb E_x^\pi\sum_{t=0}^{\infty}
\gamma^t\Delta_{\calP}(X_t,A_t)\right|\nonumber\\
&\le\mathbb E_x^\pi\sum_{t=0}^{\infty}
\gamma^t|\Delta_{\calP}(X_t,A_t)|\nonumber\\
&\le\sum_{\substack{i<j:\\ \calP(i)\ne\calP(j)}}
\lambda_{ij}\sum_{t=\tau_{ij}(x)}^{\infty}\gamma^t
=B_{\calP}(x).
\end{align}
Bounded rewards and $0<\gamma<1$ justify exchanging the finite pair sum,
expectation, and discounted series.
\end{proof}

The proof also covers a reward-blind history-dependent action kernel if exactly the
same kernel is held fixed across models.  It does not cover a controller that reads
deleted rewards and changes future actions: then the coupled state--action paths can
diverge even though the physical kernel is unchanged.

\subsection{Cut-MDP factorization}

\begin{proof}[Expanded proof of Proposition~\ref{prop:factorization}]
For each $C\in\calP$, define
\begin{align}
r_C(x_C,a_C)&=\sum_{i\in C}r_i(x_i,a_i)
+\sum_{\substack{i,j\in C:\\i\ne j}}
r_{ij}(x_i,x_j,a_i,a_j),\\
K_C(x_C'\mid x_C,a_C)&=\prod_{i\in C}K_i(x_i'\mid x_i,a_i).
\end{align}
Then $r_{\calP}=\sum_Cr_C$ and $K=\prod_CK_C$.  Let $V_C^*$ be the unique
discounted Bellman fixed point for team $C$ and set
$\overline V(x)=\sum_CV_C^*(x_C)$.  Applying the centralized cut-MDP Bellman
operator yields
\begin{align}
(T_{\calP}\overline V)(x)
&=\max_{a\in\prod_C\mathcal A_C}
\sum_C\left[r_C(x_C,a_C)+
\gamma\mathbb E_{K_C}V_C^*(X_C')\right]\nonumber\\
&=\sum_C\max_{a_C\in\mathcal A_C}
\left[r_C(x_C,a_C)+
\gamma\mathbb E_{K_C}V_C^*(X_C')\right]\nonumber\\
&=\sum_CV_C^*(x_C)=\overline V(x).
\end{align}
The discounted Bellman operator is a contraction, hence
$\overline V=V_{\MP}^*$.  Selecting a maximizing action separately in each team
gives a deterministic Markov product policy attaining this value.
\end{proof}

This separation needs Cartesian team action spaces, product transitions, and the
absence of all cross-team reward terms.  Shared action constraints, coupled
dynamics, or undeleted higher-order rewards crossing teams invalidate the exchange
of the joint maximum and team sum.

\subsection{Two-sided transfer and tight constants}

\begin{proof}[Expanded proof of Corollary~\ref{cor:regret}]
A finite discounted MDP admits a stationary Markov optimum $\pi^*$.  Insert cut-MDP
values:
\begin{align}
V_{\M}^{\pi^*}-V_{\M}^{\widehat\pi_{\calP}}
={}&(V_{\M}^{\pi^*}-V_{\MP}^{\pi^*})\nonumber\\
&+(V_{\MP}^{\pi^*}-V_{\MP}^{\widehat\pi_{\calP}})\nonumber\\
&+(V_{\MP}^{\widehat\pi_{\calP}}-
V_{\M}^{\widehat\pi_{\calP}}).
\end{align}
Cut optimality makes the middle term nonpositive.  Theorem~\ref{thm:policy}
upper-bounds each remaining signed term by $B_{\calP}(x)$, proving the result.
\end{proof}

\paragraph{Tight fixed-policy constant.}
Use two singleton teams in a one-state MDP, $\tau_{12}=0$, and one constant deleted
ordered reward $b>0$.  There is a unique policy, and the value difference is
$b/(1-\gamma)=B_{\calP}$.

\paragraph{Tight factor two.}
Again use two singleton teams and one state, now with two actions per agent.  All
local and cut rewards are zero.  The deleted reward is $+b$ at joint action $(1,1)$,
$-b$ at $(0,0)$, and zero otherwise.  Both are cut-optimal product policies.  If
the cut solver returns the $(0,0)$ tie while the original optimum selects $(1,1)$,
their original values are $-b/(1-\gamma)$ and $+b/(1-\gamma)$.  The regret equals
$2b/(1-\gamma)=2B_{\calP}$.  Thus no smaller universal constant is possible without
an additional sign, tie-breaking, or structural condition.  This worst-case witness
does not imply that typical instances approach the bound.

\section{Algorithms and Correctness}
\label{app:algorithms}

\subsection{RCSD algorithm}

The following pseudocode maps the certificate construction to partitioned planning.

\begin{enumerate}
  \item \textbf{Input:} initial state $x$, capacity $L$, discount $\gamma$, valid
  bounds $(v_i,R_{ij},b_{ij})$, and a capacity-partition solver.
  \item For every unordered pair $i<j$, set
  $\lambda_{ij}\leftarrow b_{ij}+b_{ji}$, compute $\tau_{ij}(x)$ by
  Eq.~\eqref{eq:tau}, and set
  $q_{ij}\leftarrow\lambda_{ij}\gamma^{\tau_{ij}}/(1-\gamma)$.
  \item Invoke the solver on the complete weighted graph $(q_{ij})$ and capacity
  $L$, returning a capacity-valid partition $\calP$.  The exact variant maximizes
  retained affinity; heuristic variants need not do so.
  \item For every $C\in\calP$, construct $(\mathcal X_C,\mathcal A_C,K_C,r_C)$
  and solve it for a team-optimal policy $\widehat\pi_C$.
  \item Return $\calP$ and product policy
  $\widehat\pi_{\calP}=\prod_C\widehat\pi_C$.
  \item Return the certificate
  $B_{\calP}(x)=\sum_{i<j:\calP(i)\ne\calP(j)}q_{ij}(x)$.
\end{enumerate}

Step 2 is the scientific mechanism.  Step 3 is a standard capacity-bounded
coalition optimization problem~\cite{levinger2024bounded}; replacing its solver
does not alter Theorem~\ref{thm:policy} for the returned partition.  The
$2B_{\calP}$ corollary additionally requires the team-optimal policies in Step 4;
an approximate planner retains the same-policy deletion certificate but does not
inherit that regret bound without an additional planning-error term.  Exact value
iteration is used in our finite experiments.

\subsection{Exact subset dynamic program}

For a remaining agent set $S$, fix $i=\min S$ and enumerate all blocks $C$ such
that $i\in C\subseteq S$ and $1\le|C|\le L$.  Let
$W(C)=\sum_{u<v\in C}q_{uv}$ and define
\begin{equation}
F(S)=\max_{\substack{C\subseteq S:\;\min S\in C\\1\le|C|\le L}}
\{W(C)+F(S\setminus C)\},\qquad F(\varnothing)=0.
\label{eq:subset-dp}
\end{equation}
Memoize the maximizing block and recurrence value.  Every feasible partition of a
nonempty $S$ contains exactly one block containing $\min S$; enumerating that block
and recursing lists every feasible partition exactly once at the recurrence level.
Additivity gives optimal substructure, and induction on $|S|$ proves exactness.
The implementation caches values and block backpointers and compares objectives
at absolute tolerance $10^{-12}$.  Within tolerance, ties prefer more internal
links and then the lexicographic partition; ``exact'' is relative to this tolerance.

For fixed $L$, there are $2^n$ subset states and at most
$\sum_{k=0}^{L-1}\binom{n-1}{k}=O(n^{L-1})$ candidate blocks per state.  Direct
block-affinity precomputation gives the conservative bounds reported in the main
paper.  The exact method is used at oracle sizes only.

\subsection{Direct feasible greedy merge}

\begin{enumerate}
  \item Initialize $\calP\leftarrow\{\{1\},\ldots,\{n\}\}$.
  \item Among pairs $A,B\in\calP$ satisfying $|A|+|B|\le L$, compute merge gain
  $G(A,B)=\sum_{i\in A,j\in B}q_{ij}$.
  \item If no feasible pair exists or the best gain is nonpositive, return
  $\calP$.  Otherwise replace $A,B$ by $A\cup B$ and repeat Step 2.
\end{enumerate}

Every merge preserves capacity, and $G(A,B)$ is exactly the cut certificate rescued
when the $A$--$B$ edges become internal.  Greedy therefore monotonically decreases
$B_{\calP}$, but monotonicity supplies no global approximation ratio.  Canonical
ties are resolved by the resulting directed link count and lexicographic partition.

\subsection{MnM-sum implementation}

Our baseline first computes a deterministic maximum weight matching among
singletons.  Each later round matches newly grown coalitions to remaining
singletons.  A contracted edge sums all cross-coalition $q_{ij}$ values.  Matched
coalitions grow by one; unmatched grown coalitions remain unchanged.  The process stops at
capacity or when no positive match remains.  ``MnM-sum'' records this summed
contraction because the published pseudocode does not uniquely prescribe weighted
contracted edges~\cite{levinger2024bounded}.  We treat it only as an implementation
baseline and do not transfer the paper's approximation statement to this completion.

\section{Finite-MDP Verification and Boundary Cases}
\label{app:validation}

The verification suite covers finite discounted MDPs with product state/action
spaces, independent bounded-motion transitions, pre-transition local and ordered-
pair rewards, fixed partitions, and stationary Markov policies.  It checks the
implementation against the analytical results and constructs boundary examples;
it is not a substitute for the proofs.  Learning, partial observation, dynamic
repartitioning, and LIMDP visibility policies are not evaluated.

\begin{table}[h]
\caption{Finite-MDP verification measurements and outcomes.}
\label{tab:theorem-validation}
\centering
\small
\begin{tabular}{@{}p{0.36\columnwidth}p{0.27\columnwidth}p{0.22\columnwidth}@{}}
\toprule
Measurement & Quantity evaluated & Result\\
\midrule
Pre-contact reachability & 648 settings & 0 early contacts\\
Same-policy value & maximum gap$-B$ & $-0.3547$\\
Extremal reward oracle & maximum gap$-B$ & $-0.1906$\\
Cut vs. product optimum & maximum value difference & 0.0\\
Execution regret & maximum regret$-2B$ & $-1.4287$\\
Numerical residual & maximum Bellman residual & $1.42\!\times\!10^{-14}$\\
Invalid-scope controls & excess over invalid bound & 1.0, 2.0, 3.0\\
Tight $B$ / $2B$ examples & observed/bound ratio & 1.0 / 1.0\\
Positive-regret example & centralized regret & 1.0\\
\bottomrule
\end{tabular}
\end{table}

\paragraph{Reachability enumeration.}
The sweep exhausts two-agent one-dimensional grids of widths two through five,
ranges zero through two, speed indicators zero or one, every initial state, and
all feasible action sequences through horizon six.  It expands 124,380 transitions
and checks 23,268 reachable states.  Sixty-eight cases have zero relative speed
outside range and correctly return infinite contact time and zero affinity.

\paragraph{Fixed-policy enumeration.}
The validation evaluates all 256 stationary joint policies for each of eight asymmetric
signed reward tables, totaling 2,048 policy evaluations.  It adds 200 seeded
randomized policies over 40 stochastic independent-motion instances, all with
nonzero slip.  Policy evaluation solves $(I-\gamma P_\pi)V=r_\pi$ rather than
sampling rollouts.  For deleted rewards it also solves the reward and its negation,
computing an extremal Markov-policy gap instead of relying only on sampled policies.

\paragraph{Optimization enumeration.}
Eighteen two-, three-, and four-agent fixtures generate 96 capacity-feasible
partitions and 1,716 statewise product-optimum/regret comparisons.  The centralized
cut optimum is compared with independently solved team optima; the product policy
is then evaluated under both cut and original rewards.

\paragraph{Verification procedure.}
Bounded grid trajectories establish the absence of range entry before
$\tau_{ij}$.  Original and cut values are then solved for every signed reward
table and stationary deterministic policy, followed by stochastic kernels and
randomized Markov policies.  For each optimization fixture and partition, we
compare the centralized cut optimum with independently solved team optima and
evaluate the product policy in both reward models.  Boundary constructions vary
reward timing, directed-envelope availability, and policy class, while separate
witnesses attain $B$ and $2B$.  Counts, extrema, residuals, and seeds are retained
with the numerical results.

\subsection{Counterexamples Beyond the Assumptions}

\begin{table}[h]
\caption{Counterexamples obtained after removing individual assumptions.}
\label{tab:assumption-counterexamples}
\centering
\small
\begin{tabular}{lrrr}
\toprule
Invalid extension & Gap & Wrong bound & Scoped bound\\
\midrule
Post-transition reward & 2.0 & 1.0 & 2.0\\
One directed envelope & 2.2 & 0.2 & 2.2\\
Reward-history reaction & 5.0 & 2.0 & 1.0$^\dagger$\\
\bottomrule
\end{tabular}
\vspace{1mm}
\parbox{0.96\columnwidth}{\footnotesize $^\dagger$The final number is the gap
for the fixed Markov comparator, not a bound on the reward-reactive policy.}
\end{table}

The timing example demonstrates an off-by-one: if reward
is defined on $X_{t+1}$, first contact at state time $\tau$ can contribute at return
index $\tau-1$.  The directed example enforces the ordered reward convention of
Eq.~\eqref{eq:reward-factorization}.  The history-dependent example shows that
Theorem~\ref{thm:policy} applies to a fixed policy rather than a
reward-observing controller.

\section{Meeting-Port Corridor Protocol}
\label{app:corridor}

\subsection{Generator and exact evaluation}

Positions are $x_i\in\{0,\ldots,6\}$, actions are
$u_i\in\{-1,0,+1\}$, and
$x_i'=\operatorname{clip}(x_i+u_i,0,6)$.  Local reward is
\begin{equation}
r_i(x_i,u_i)=-c_i|x_i-h_i|-\eta_i^-\Ind\{u_i=-1\}
-\eta_i^+\Ind\{u_i=+1\}.
\end{equation}
For a meeting port $m_{ij}$,
\begin{equation}
r_{ij}+r_{ji}=\lambda_{ij}
\Ind\{x_i=x_j=m_{ij}\}\Ind\{u_i=u_j=0\},
\end{equation}
The implementation stores this sum as one aggregate pair term; equivalently, each
ordered direction may be assigned $\lambda_{ij}/2$.  We draw
$c_i\sim U[0.025,0.050]$ and
$\eta_i^-,\eta_i^+\sim U[0.006,0.014]$.  The exact generator consumes random
draws in the order: five $c_i$, five $\eta_i^-$, five $\eta_i^+$, ten envelopes in
lexicographic pair order, then family-specific port randomness.  PCG64 is seeded
directly by the stored integer.

The informative family uses homes $(0,1,3,5,6)$ and envelopes
$U[0.45,1.35]$.  Each pair's port is a Bernoulli choice between the floor and
ceiling midpoint, so first possible port contact equals
$\tau_{ij}=\lceil|h_i-h_j|/2\rceil$.  In the omitted-feature control, all homes
are at position 3 and envelopes follow $U[0.78,0.82]$.  Its ports are a seeded
permutation of the multiset
\begin{equation*}
\{0,0,1,1,1,5,5,5,6,6\}.
\end{equation*}
This makes delay uninformative and hides incompatible-port structure from RCSD.

For a partition, each team MDP is solved independently.  The team policies form a
global product action.  Deterministic execution from the home state eventually
revisits a state; the infinite return is evaluated exactly as a finite discounted
prefix plus a geometric cycle.  The cut-return evaluation must agree with the sum
of team initial values within $2\times10^{-8}$; the observed maximum discrepancy is
$1.33\times10^{-9}$.

\subsection{Stratum-evaluation pseudocode}

\begin{enumerate}
  \item Solve the full-reward centralized MDP by value iteration.
  \item Solve every agent subset of size at most $L$ as a team MDP.
  \item Enumerate every canonical partition with the fixed equal-
  communication shape.
  \item Compute all $q_{ij}$ values and select exact RCSD, distance, envelope, current,
  uniform, and oracle partitions using only each selector's specified inputs.
  \item For each partition, assemble the product team policy; evaluate original
  and cut returns by prefix--cycle summation; save partition, $B_{\calP}$,
  centralized regret, normalized regret, residuals, and directed links.
  \item Evaluate capacity-valid underfilled solver/Pareto partitions separately,
  never mixing them into equal-shape primary tests.
\end{enumerate}

\begin{table*}[t]
\caption{Complete partition-mechanism specification.  Only the first four rows
enter the primary equal-shape comparisons.}
\label{tab:baselines}
\centering
\small
\begin{tabular}{llll}
\toprule
Mechanism & Information used & Partition support & Role\\
\midrule
RCSD (exact) & $\lambda_{ij},\tau_{ij},\gamma$ & all equal-shape partitions & proposed\\
Distance only & initial distance & same set & primary single-factor\\
Envelope only & $\lambda_{ij}$ & same set & primary single-factor\\
Uniform & none & exact mean over same set & primary random\\
Current contact & current in-range envelope & same set & diagnostic\\
Hindsight oracle & realized regret & same set & diagnostic only\\
RCSD (greedy) & RCSD affinity & all capacity-valid & scalable solver\\
RCSD (MnM-sum) & RCSD affinity & all capacity-valid & solver baseline\\
Visibility components & initial-distance components & capacity-valid if available & secondary\\
\bottomrule
\end{tabular}
\end{table*}

For the informative family, the visibility-components selector is feasible in every stratum.
At $L=2$ it has shape $2+2+1$;
at $L=3$ it remains $2+2+1$, creating 4,500 auxiliary rows outside the primary
$3+2$ support.  It is infeasible on all omitted-feature-control strata because the
initially coincident agents form a component larger than capacity, so it is omitted
for those strata.

\section{Evaluation Protocol and Complete Results}
\label{app:evaluation-results}

Numerical-accuracy and reported evaluations use fixed disjoint seed sets;
the exact integers are provided with the supplementary material.  At
$L=2$, shape $2+2+1$ has 15 partitions and four directed internal links; at $L=3$,
shape $3+2$ has ten partitions and eight links.  The resulting 18,000 strata
contain 225,000 complete equal-shape rows plus 4,500 visibility-components
auxiliary rows.

\begin{table*}[t]
\caption{Complete evaluation summary.  Utility effects average $L=2,3$ within
each of 4,500 instances; rank statistics use 9,000 instance/capacity strata.}
\label{tab:evaluation-summary}
\centering
\small
\renewcommand{\arraystretch}{0.98}
\begin{tabular}{lll}
\toprule
Evaluation & Metric & Result\\
\midrule
Informative utility & Uniform ARR / win / $d_z$ / Holm $p$ & .5604/.9709/1.7911/.000030\\
Informative utility & Distance ARR / win; Envelope ARR / win & .2880/.7353; .2531/.7323\\
Certificate ranking & median Spearman / positive strata & .7143/.9757\\
Certificate ranking & clustered 95\% interval & [.7091,.7212]\\
Greedy objective quality & median/p90/max normalized gap & .01110/.05287/.13845\\
100-agent construction & MnM-sum median/p95 at $L=4$ & .2333/.2469 s\\
Omitted-feature control & absolute / signed median Spearman & .2929/.1643\\
Omitted-feature control & ARR / win / paired 95\% interval & .03096/.5164/[.00155,.00281]\\
\bottomrule
\end{tabular}
\end{table*}

\subsection{Paired statistical procedure}

For each evaluation seed $s$, let $r_s$ be RCSD-Exact normalized regret averaged over
$L=2,3$, and let $q_s$ be the corresponding baseline average.  Set
$d_s=q_s-r_s$.  The analysis computes
\begin{enumerate}
  \item $\mathrm{ARR}=1-\sum_sr_s/\sum_sq_s$;
  \item paired standardized effect $d_z=\overline d/\operatorname{sd}(d)$;
  \item win rate $\operatorname{mean}(\Ind\{d>10^{-10}\}+
  \tfrac12\Ind\{|d|\le10^{-10}\})$;
  \item 10,000 paired seed bootstrap resamples;
  \item 100,000 paired sign-flips, followed by Holm adjustment
  of Uniform, Distance, and Envelope $p$-values as one family.
\end{enumerate}
For certificate ranking, bootstrap draws resample instance seeds and carry both capacities as one
cluster.  Thus the 9,000 rank correlations per family are not treated as independent
for the interval.

\subsection{Full utility statistics}

\begin{table}[ht]
\caption{Informative-family exact utility comparison ($n=4{,}500$ paired seeds).}
\label{tab:full-utility}
\centering
\small
\begin{tabular}{lrrrr}
\toprule
Comparator & ARR & Mean diff. & $d_z$ & Win\\
\midrule
Uniform & .56042 & .04985 & 1.7911 & .9709\\
Distance & .28796 & .01581 & .6039 & .7353\\
Envelope & .25313 & .01325 & .5541 & .7323\\
Current & .31315 & .01783 & .7083 & .7777\\
\bottomrule
\end{tabular}
\end{table}

The 95\% paired-reduction intervals and Holm-adjusted one-sided sign-flip tests are
\begin{align*}
\mathrm{CI}_{U}&=[0.04905,0.05068], &p_U&=2.99997\!\times\!10^{-5},\\
\mathrm{CI}_{D}&=[0.01504,0.01656], &p_D&=2.99997\!\times\!10^{-5},\\
\mathrm{CI}_{E}&=[0.01255,0.01396], &p_E&=2.99997\!\times\!10^{-5}.
\end{align*}
All 4,500 informative-family seeds have nonzero aggregate Uniform comparator
regret.  Of 107,833 partition rows with positive regret, 97.3\% have
$\mathrm{regret}/(2B)\ge .01$.  This indicates that the certificate is numerically
non-negligible on most tested rows; it is not evidence that the bound is generally
tight.

\subsection{Exact solver quality}

Table~\ref{tab:solver-quality-main} reports every exact-oracle cell.  The zero
RCSD-MnM-sum gaps at $L=2$ reflect singleton maximum-weight matching on the tested
positive even-$n$ graphs and do not extend to larger capacities.  At $L=4$,
RCSD-MnM-sum uses one third of the permitted directed communication, explaining
its larger objective gaps.  Greedy and dense positive-weight mechanisms fill
capacity; Current may underfill because only strictly positive edges are merged.

\subsection{Scaling details}

The scaling generator draws positions uniformly in $[0,10]^2$, speeds in
$[0.5,1.5]$, symmetric ranges in $[0.5,2]$, aggregate envelopes log-uniformly in
$[0.25,4]$, and uses $\gamma=0.9$.  There are 300 seeds in every cell.

\begin{table}[t]
\caption{Partition-only scaling.  Entries are median / empirical 95th percentile
seconds; Comm. is Current's median [range] fraction of the $n(L-1)$ ceiling.}
\label{tab:scaling-full}
\centering
\scriptsize
\begin{tabular}{rrlll}
\toprule
$n$ & $L$ & Greedy & MnM-sum & Comm.\\
\midrule
24 & 2 & .00065 / .00074 & .00374 / .00442 & .583 [.250,.833]\\
24 & 4 & .00094 / .00103 & .00377 / .00440 & .444 [.139,.722]\\
48 & 2 & .00455 / .00474 & .02695 / .02946 & .750 [.542,.875]\\
48 & 4 & .00647 / .00673 & .02699 / .02951 & .667 [.500,.889]\\
100 & 2 & .03822 / .04009 & .23408 / .24885 & .860 [.740,.920]\\
100 & 4 & .05453 / .05694 & .23326 / .24689 & .833 [.760,.913]\\
\bottomrule
\end{tabular}
\end{table}

These empirical timings characterize the implementation rather than asymptotic
complexity.  No MDP planning occurs in this arm.

\section{Omitted-Feature Control Analysis}
\label{app:control}

Because every initial distance and delay is zero, the RCSD-Exact, envelope-only, and
current-contact selectors are algebraically identical; all three select the same
partition in all 9,000 control strata.  The large-sample ARR is .03096 with
$d_z=.1000$ and win rate .5164: precisely estimated, but much smaller than the
informative-family effects.  The signed median rank association is .1643 and the
median absolute magnitude is .2929.  These residual associations show that this
control does not establish equivalence or exact specificity.

On an exploratory 60-stratum subset, recomputing every selection
from envelopes alone, without reading ports, reproduces 60/60 selections.  No
exact objective tie occurs at tolerance $10^{-12}$; the smallest non-tied top-two
margin is 0.001201.

Two post hoc analyses on that subset are reported as exploratory.  A naive
row-permutation null places the observed median absolute rank correlation above
its 95\% interval $[.1543,.2732]$ (upper-tail $p=.00090$), but ignores that partitions share
edges.  A mechanism-aware null reassigns the ten observed envelopes to edges within
each seed, uses the same reassignment at both capacities, and preserves the
partition/regret incidence structure.  Its 95\% interval is $[.1821,.3375]$.
It contains the observation (upper-tail $p=.0774$).  Direct recomputation confirms that the
selector uses only envelopes, and the margin analysis rules out an exact-tie explanation.
These analyses do not turn the control into an equivalence test or identify the
omitted feature as the sole causal difference between generators.

\section{Random Two-Dimensional End-to-End Evaluation}
\label{app:twod}

\subsection{Task and estimands}

Both tiers use connected grids with static obstacles, fixed subteams, independent
agent motion, and pre-transition rewards.  Agents may co-locate and cross, so
collisions do not couple the transition kernel.  A failed move leaves an agent in
place; otherwise it follows its selected local move.  Unary rewards combine
service-site preferences and movement costs.  Pair rewards are positive only
inside a symmetric finite interaction range and a seeded pair-specific rendezvous
region.  Cross-team pair rewards are removed from the planning surrogate but
remain present when the returned policy is evaluated in the original task.

The small tier uses a $3\times3$ grid with one obstacle, four agents, five local
actions, $\gamma=.85$, speeds in $\{1,2\}$, independent stay-slip probabilities in
$[.08,.25]$, and pair ranges in $\{0,1\}$.  Capacity $L=2$ gives exactly three
$2+2$ partitions.  For each of 128 maps, value iteration solves the unrestricted
centralized stationary-Markov MDP, every two-agent team MDP, and all three product
team policies.  Its reported regret is therefore against the unrestricted
stationary-Markov optimum, not against a hand-designed controller set.

The larger tier uses a connected $7\times7$ grid with 10--22\% obstacles, four
service sites, $n\in\{8,12,16,20\}$, and $L\in\{2,4\}$.  Speeds, slip
probabilities, local rewards, ranges in $\{0,1,2\}$, and both directed reward
envelopes are heterogeneous.  Four stationary shortest-path waypoint controllers
are generated before partition selection and shared by every method.  Exact
finite-prefix marginal occupancy propagation evaluates their unary and pairwise
returns; a target-absorbing tail approximation is continued until each utility
term has error below $10^{-9}$.  Mixed-integer optimization then finds the best
controller assignment centrally and within each selected team.  Incumbent--dual
gaps and accumulated tail errors give the finite-precision audit
\begin{equation}
\operatorname{Reg}_{\mathrm{lib}}(\calP)
\le 2B_{\calP}+\epsilon_{\mathrm{num}}.
\end{equation}
This tier measures regret only within the fixed waypoint-controller library.

\subsection{Comparators and statistics}

RCSD-MIP, Distance-MIP, Envelope-MIP, and Current-MIP optimize the same complete
block shape using, respectively, $q_{ij}$, $1/(1+D_{ij})$,
$\lambda_{ij}/(1-\gamma)$, and the current-contact envelope.  Uniform samples the
same shape.  Value-MIP scores coalitions with their optimal cut-library value and
is an information-richer planning-aware comparator.  Every method therefore uses
the same $n(L-1)$ directed persistent links; online message traffic is not
measured.

Generated maps are the statistical units: 128 paired seeds in the small tier and
30 maps at each $n$ in the large tier, where $L=2,4$ is averaged within map and
the paired bootstrap is stratified by $n$.  Confirmatory large-tier differences
use raw controller-library regret.  For cross-task display, normalized regret is
$g/Z$, where
$Z=[\sum_{i<j}\lambda_{ij}+\sum_i\max_k u_{ik}]/(1-\gamma)$ and $u_{ik}$ is
agent $i$'s service reward at site $k$.  Figure~\ref{fig:twod}a shows these
per-map normalized differences.  All intervals are per-comparator 95\% intervals,
not simultaneous family-wise statements.

\begin{table*}[t]
\caption{Random 2-D selector utility.  Normalized regret is scaled within the
generated task.  Small ARR is the aggregate raw-regret reduction achieved by RCSD
relative to each comparator.
Large $\Delta$ is raw comparator regret minus RCSD regret after averaging
capacities; positive values favor RCSD, and intervals use an $n$-stratified map
bootstrap.  Value-MIP is available only in the fixed-controller tier.}
\label{tab:twod-utility}
\centering
\small
\begin{tabular}{lrrrr}
\toprule
Method & Small norm. regret & Small ARR & Large norm. regret & Large $\Delta$ [95\% CI]\\
\midrule
RCSD & .0826 & -- & .0963 & --\\
Uniform & .1070 & .216 & .1005 & $-1.89$ [$-8.83$, $4.78$]\\
Distance & .1110 & .241 & .1164 & $10.05$ [$5.36$, $14.97$]\\
Envelope & .0850 & .026 & .0876 & $-5.80$ [$-11.23$, $-.57$]\\
Current & .0978 & .148 & .1082 & $4.90$ [$.91$, $8.81$]\\
Value-MIP & -- & -- & .0879 & $-5.44$ [$-8.66$, $-2.30$]\\
\bottomrule
\end{tabular}
\end{table*}

\begin{figure*}[t]
  \centering
  \includegraphics[width=.315\textwidth]{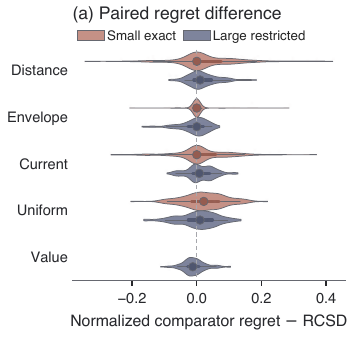}\hfill
  \includegraphics[width=.292\textwidth]{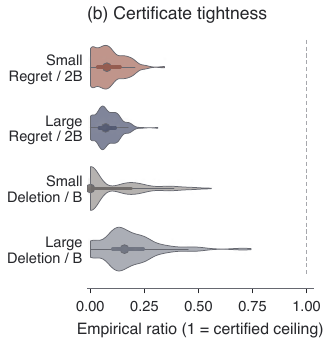}\hfill
  \includegraphics[width=.355\textwidth]{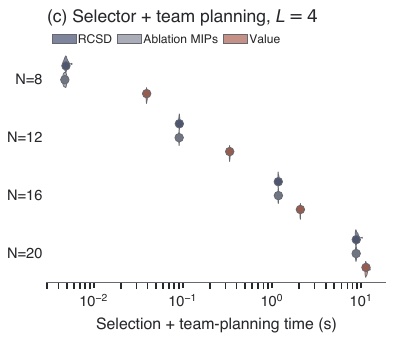}
  \caption{\textbf{Random two-dimensional end-to-end evaluation.}
  (a) Seed-paired normalized regret differences $g_Q-g_{\mathrm{RCSD}}$;
  positive values favor RCSD.  Small exact retains 128 four-agent MDPs; large
  restricted retains 120 map units and averages $L=2,4$ within each map.
  Value-MIP uses planning information unavailable to certificate-input selectors.
  (b) RCSD fixed-policy deletion error divided by $B_{\calP}$ and execution
  regret divided by $2B_{\calP}$; large-tier regret is controller-library
  restricted, and the dashed line is the certified ceiling.  (c) At $L=4$,
  precomputed-score aggregation, fixed-shape partitioning, and selected-team
  planning time; controller-utility and pair-feature construction are excluded.
  Ablation MIPs is the within-map median of Distance, Envelope, and Current.
  Violins retain all plotted observations; inference in (a) uses
  maps as the unit.  Points and thick/thin intervals denote medians,
  interquartile ranges, and 5th--95th percentiles.}
  \Description{Three distribution panels compare paired normalized regret
  differences, empirical certificate ratios against their theoretical ceiling,
  and selection-plus-team-planning runtime as the number of agents grows.}
  \label{fig:twod}
\end{figure*}

\subsection{Certificate tightness and utility boundary}

\begin{table}[t]
\caption{Empirical RCSD certificate ratios (median/p90/p95).  Quantiles use point
values for RCSD-selected partitions; violation counts use the numerically enlarged
upper interval and cover all evaluated partitions or method rows.}
\label{tab:twod-tightness}
\centering
\scriptsize
\begin{tabular}{@{}lccr@{}}
\toprule
Tier & $|\Delta V|/B$ & Regret/$2B$ & Viol.\\
\midrule
Small exact & .0001/.323/.387 & .076/.184/.206 & 0/384\\
Large library & .158/.341/.452 & .070/.149/.175 & 0/1,440\\
\bottomrule
\end{tabular}
\end{table}

All 384 small-tier partition evaluations have positive $B_{\calP}$, and 366 have
positive unrestricted regret.  RCSD has lower aggregate regret than Uniform,
Distance, and Current, but paired raw evidence is borderline for Current; its
difference from Envelope is unresolved.  The Current raw interval has a .0026
lower endpoint under the original 10,000-draw bootstrap and becomes
$[-.00035,.83245]$ in a one-million-draw Monte Carlo sensitivity check.  In the
large tier, pooled pointwise
intervals favor RCSD over Distance and Current, include zero for Uniform, and
favor Envelope and Value-MIP over RCSD.  The normalized paired analysis has the
same interval conclusions.  One plausible task-specific explanation is that
the fixed waypoint library makes long-run service location and reward magnitude
more predictive of realized value than the initial speed-limit delay.  This does
not affect the conditional certificate: it shows that minimizing a worst-case
upper bound need not minimize typical regret.

The ratios in Table~\ref{tab:twod-tightness} are well below one, so the certificate
is conservative rather than calibrated to realized loss.  Its upper tail is not
vacuous: the largest audited ratios over all large-tier methods are .816 for
deletion and .356 for restricted regret.  Every map is connected, every reward and
speed audit passes, all mixed-integer gaps are zero at the solver tolerance, and
the maximum propagated regret-interval width is $7.74\times10^{-7}$.  At $n=20$,
the complete six-method, two-capacity suite has median/95th-percentile runtime
50.16/52.34 seconds; this is an implementation observation, not a deployment
latency guarantee.

\section{Reproducibility Details}
\label{app:reproduction}

Reproduction supports Python 3.12 or later; reported outputs use Python 3.13.1,
NumPy 1.26.4, SciPy 1.15.2, NetworkX 3.6.1, and seaborn
\cite{harris2020numpy,virtanen2020scipy,hagberg2008networkx,waskom2021seaborn}.
The supplementary material contains implementation, fixed configurations, and
seed-level outputs; no learned model, external service, or network call is needed.
Checks cover contact and value bounds, exact five-agent evaluation,
approximate-versus-exact minimization, construction scaling, both 2-D tiers,
complete seed sets and capacities, Bellman and cut-policy residuals, regret
nonnegativity to $10^{-8}$, finite-tail and solver intervals, and
Eq.~\eqref{eq:regret-bound}.

%% file: references.bib
@inproceedings{deweese2024limdp,
  author    = {Deweese, Alex and Qu, Guannan},
  title     = {Locally Interdependent Multi-Agent {MDP}: Theoretical Framework for Decentralized Agents with Dynamic Dependencies},
  booktitle = {Proceedings of the 41st International Conference on Machine Learning},
  series    = {Proceedings of Machine Learning Research},
  volume    = {235},
  pages     = {10687--10709},
  year      = {2024},
  publisher = {PMLR},
  url       = {https://proceedings.mlr.press/v235/deweese24a.html}
}

@inproceedings{qu2020localized,
  author    = {Qu, Guannan and Wierman, Adam and Li, Na},
  title     = {Scalable Reinforcement Learning of Localized Policies for Multi-Agent Networked Systems},
  booktitle = {Proceedings of the 2nd Conference on Learning for Dynamics and Control},
  series    = {Proceedings of Machine Learning Research},
  volume    = {120},
  pages     = {256--266},
  year      = {2020},
  publisher = {PMLR},
  url       = {https://proceedings.mlr.press/v120/qu20a.html}
}

@article{becker2004transition,
  author  = {Becker, Raphen and Zilberstein, Shlomo and Lesser, Victor R. and Goldman, Claudia V.},
  title   = {Solving Transition Independent Decentralized {M}arkov Decision Processes},
  journal = {Journal of Artificial Intelligence Research},
  volume  = {22},
  pages   = {423--455},
  year    = {2004},
  doi     = {10.1613/jair.1497},
  url     = {https://auld.aaai.org/Library/JAIR/Vol22/jair22-013.php}
}

@inproceedings{scharpff2016transition,
  author    = {Scharpff, Joris and Roijers, Diederik M. and Oliehoek, Frans A. and Spaan, Matthijs T. J. and de Weerdt, Mathijs M.},
  title     = {Solving Transition-Independent Multi-Agent {MDP}s with Sparse Interactions},
  booktitle = {Proceedings of the Thirtieth AAAI Conference on Artificial Intelligence},
  year      = {2016},
  pages     = {3174--3180},
  doi       = {10.1609/aaai.v30i1.10405},
  url       = {https://ojs.aaai.org/index.php/AAAI/article/view/10405}
}

@inproceedings{guestrin2001factored,
  author    = {Guestrin, Carlos and Koller, Daphne and Parr, Ronald},
  title     = {Multiagent Planning with Factored {MDP}s},
  booktitle = {Advances in Neural Information Processing Systems 14},
  pages     = {1523--1530},
  year      = {2001},
  publisher = {MIT Press},
  url       = {https://proceedings.neurips.cc/paper_files/paper/2001/hash/7af6266cc52234b5aa339b16695f7fc4-Abstract.html}
}

@article{fiscko2025clustered,
  author  = {Fiscko, Carmel and Kar, Soummya and Sinopoli, Bruno},
  title   = {Clustered Control of Transition-Independent {MDP}s},
  journal = {IEEE Transactions on Control of Network Systems},
  volume  = {12},
  number  = {3},
  pages   = {1881--1893},
  year    = {2025},
  doi     = {10.1109/TCNS.2023.3330925},
  url     = {https://doi.org/10.1109/TCNS.2023.3330925}
}

@inproceedings{fiscko2023reachability,
  author    = {Fiscko, Carmel and Kar, Soummya and Sinopoli, Bruno},
  title     = {Maximizing Reachability in Factored {MDP}s via Near-Optimal Clustering with Applications to Control of Multi-Agent Systems},
  booktitle = {2023 62nd IEEE Conference on Decision and Control ({CDC})},
  pages     = {7970--7975},
  year      = {2023},
  publisher = {IEEE},
  doi       = {10.1109/CDC49753.2023.10383361},
  url       = {https://doi.org/10.1109/CDC49753.2023.10383361}
}

@article{kok2006payoff,
  author  = {Kok, Jelle R. and Vlassis, Nikos},
  title   = {Collaborative Multiagent Reinforcement Learning by Payoff Propagation},
  journal = {Journal of Machine Learning Research},
  volume  = {7},
  number  = {65},
  pages   = {1789--1828},
  year    = {2006},
  url     = {https://www.jmlr.org/papers/v7/kok06a.html}
}

@article{bernstein2002complexity,
  author  = {Bernstein, Daniel S. and Givan, Robert and Immerman, Neil and Zilberstein, Shlomo},
  title   = {The Complexity of Decentralized Control of {M}arkov Decision Processes},
  journal = {Mathematics of Operations Research},
  volume  = {27},
  number  = {4},
  pages   = {819--840},
  year    = {2002},
  doi     = {10.1287/moor.27.4.819.297}
}

@article{kearns2002simulation,
  author  = {Kearns, Michael J. and Singh, Satinder P.},
  title   = {Near-Optimal Reinforcement Learning in Polynomial Time},
  journal = {Machine Learning},
  volume  = {49},
  number  = {2--3},
  pages   = {209--232},
  year    = {2002},
  doi     = {10.1023/A:1017984413808}
}

@book{puterman1994mdp,
  author    = {Puterman, Martin L.},
  title     = {Markov Decision Processes: Discrete Stochastic Dynamic Programming},
  publisher = {Wiley},
  year      = {1994},
  doi       = {10.1002/9780470316887}
}

@book{oliehoek2016decpos,
  author    = {Oliehoek, Frans A. and Amato, Christopher},
  title     = {A Concise Introduction to Decentralized {POMDP}s},
  series    = {SpringerBriefs in Intelligent Systems},
  publisher = {Springer Cham},
  year      = {2016},
  doi       = {10.1007/978-3-319-28929-8},
  url       = {https://link.springer.com/book/10.1007/978-3-319-28929-8}
}

@inproceedings{foerster2016dial,
  author    = {Foerster, Jakob N. and Assael, Ioannis Alexandros and de Freitas, Nando and Whiteson, Shimon},
  title     = {Learning to Communicate with Deep Multi-Agent Reinforcement Learning},
  booktitle = {Advances in Neural Information Processing Systems 29},
  pages     = {2137--2145},
  year      = {2016},
  url       = {https://proceedings.neurips.cc/paper_files/paper/2016/hash/c7635bfd99248a2cdef8249ef7bfbef4-Abstract.html}
}

@inproceedings{sukhbaatar2016commnet,
  author    = {Sukhbaatar, Sainbayar and Szlam, Arthur and Fergus, Rob},
  title     = {Learning Multiagent Communication with Backpropagation},
  booktitle = {Advances in Neural Information Processing Systems 29},
  pages     = {2244--2252},
  year      = {2016},
  url       = {https://proceedings.neurips.cc/paper_files/paper/2016/hash/55b1927fdafef39c48e5b73b5d61ea60-Abstract.html}
}

@inproceedings{lowe2017maddpg,
  author    = {Lowe, Ryan and Wu, Yi and Tamar, Aviv and Harb, Jean and Abbeel, Pieter and Mordatch, Igor},
  title     = {Multi-Agent Actor-Critic for Mixed Cooperative-Competitive Environments},
  booktitle = {Advances in Neural Information Processing Systems 30},
  pages     = {6379--6390},
  year      = {2017},
  url       = {https://proceedings.neurips.cc/paper_files/paper/2017/hash/68a9750337a418a86fe06c1991a1d64c-Abstract.html}
}

@inproceedings{rashid2018qmix,
  author    = {Rashid, Tabish and Samvelyan, Mikayel and Schroeder, Christian and Farquhar, Gregory and Foerster, Jakob and Whiteson, Shimon},
  title     = {{QMIX}: Monotonic Value Function Factorisation for Deep Multi-Agent Reinforcement Learning},
  booktitle = {Proceedings of the 35th International Conference on Machine Learning},
  series    = {Proceedings of Machine Learning Research},
  volume    = {80},
  pages     = {4295--4304},
  year      = {2018},
  publisher = {PMLR},
  url       = {https://proceedings.mlr.press/v80/rashid18a.html}
}

@inproceedings{sunehag2018vdn,
  author    = {Sunehag, Peter and Lever, Guy and Gruslys, Audrunas and Czarnecki, Wojciech Marian and Zambaldi, Vinicius and Jaderberg, Max and Lanctot, Marc and Sonnerat, Nicolas and Leibo, Joel Z. and Tuyls, Karl and Graepel, Thore},
  title     = {Value-Decomposition Networks for Cooperative Multi-Agent Learning Based on Team Reward},
  booktitle = {Proceedings of the 17th International Conference on Autonomous Agents and Multiagent Systems},
  pages     = {2085--2087},
  year      = {2018},
  publisher = {IFAAMAS},
  url       = {https://www.ifaamas.org/Proceedings/aamas2018/pdfs/p2085.pdf}
}

@inproceedings{phan2021vast,
  author    = {Phan, Thomy and Ritz, Fabian and Belzner, Lenz and Altmann, Philipp and Gabor, Thomas and Linnhoff-Popien, Claudia},
  title     = {{VAST}: Value Function Factorization with Variable Agent Sub-Teams},
  booktitle = {Advances in Neural Information Processing Systems 34},
  pages     = {24018--24032},
  year      = {2021},
  publisher = {Curran Associates, Inc.},
  url       = {https://proceedings.neurips.cc/paper/2021/hash/c97e7a5153badb6576d8939469f58336-Abstract.html}
}

@inproceedings{shao2022sog,
  author    = {Shao, Jianzhun and Lou, Zhiqiang and Zhang, Hongchang and Jiang, Yuhang and He, Shuncheng and Ji, Xiangyang},
  title     = {Self-Organized Group for Cooperative Multi-agent Reinforcement Learning},
  booktitle = {Advances in Neural Information Processing Systems 35},
  pages     = {5711--5723},
  year      = {2022},
  publisher = {Curran Associates, Inc.},
  url       = {https://proceedings.neurips.cc/paper_files/paper/2022/hash/25b040c97a75021e57100648a20b1e10-Abstract-Conference.html}
}

@inproceedings{huang2022qscan,
  author    = {Huang, Wenhan and Li, Kai and Shao, Kun and Zhou, Tianze and Taylor, Matthew E. and Luo, Jun and Wang, Dongge and Mao, Hangyu and Hao, Jianye and Wang, Jun and Deng, Xiaotie},
  title     = {Multiagent {Q}-Learning with Sub-Team Coordination},
  booktitle = {Advances in Neural Information Processing Systems 35},
  pages     = {29427--29439},
  year      = {2022},
  publisher = {Curran Associates, Inc.},
  url       = {https://proceedings.neurips.cc/paper_files/paper/2022/hash/bd31bfd4caa85bffe07a35568182cdfa-Abstract-Conference.html}
}

@inproceedings{zang2023gomarl,
  author    = {Zang, Yifan and He, Jinmin and Li, Kai and Fu, Haobo and Fu, Qiang and Xing, Junliang and Cheng, Jian},
  title     = {Automatic Grouping for Efficient Cooperative Multi-Agent Reinforcement Learning},
  booktitle = {Advances in Neural Information Processing Systems 36},
  pages     = {46105--46121},
  year      = {2023},
  publisher = {Curran Associates, Inc.},
  url       = {https://proceedings.neurips.cc/paper_files/paper/2023/hash/906c860f1b7515a8ffec02dcdac74048-Abstract-Conference.html}
}

@inproceedings{liu2025hygma,
  author    = {Liu, Chiqiang and Li, Dazi},
  title     = {{HYGMA}: Hypergraph Coordination Networks with Dynamic Grouping for Multi-Agent Reinforcement Learning},
  booktitle = {Proceedings of the 42nd International Conference on Machine Learning},
  series    = {Proceedings of Machine Learning Research},
  volume    = {267},
  pages     = {38767--38788},
  year      = {2025},
  publisher = {PMLR},
  url       = {https://proceedings.mlr.press/v267/liu25af.html}
}

@inproceedings{deng2025staf,
  author    = {Deng, Zihao and Gao, Peng and Jose, Williard Joshua and Wigness, Maggie and Rogers, III, John G. and Reily, Brian and Reardon, Christopher M. and Zhang, Hao},
  title     = {Subteaming and Adaptive Formation Control for Coordinated Multi-Robot Navigation},
  booktitle = {Proceedings of The 9th Conference on Robot Learning},
  series    = {Proceedings of Machine Learning Research},
  volume    = {305},
  pages     = {2665--2677},
  year      = {2025},
  publisher = {PMLR},
  url       = {https://proceedings.mlr.press/v305/deng25b.html}
}

@inproceedings{chen2026cpo,
  author    = {Chen, Dingyang and Ye, Jianing and Zhang, Zhenyu and Kuang, Xiaolong and Shen, Xinyang and Ozer, Ozalp and Zhang, Chongjie and Zhang, Qi},
  title     = {Correlated Policy Optimization in Multi-Agent Subteams},
  booktitle = {The Fourteenth International Conference on Learning Representations},
  year      = {2026},
  url       = {https://openreview.net/forum?id=Tke3BVwUz6}
}

@inproceedings{levinger2024bounded,
  author    = {Levinger, Chaya and Hazon, Noam and Simola, Sofia and Azaria, Amos},
  title     = {Coalition Formation with Bounded Coalition Size},
  booktitle = {Proceedings of the 23rd International Conference on Autonomous Agents and Multiagent Systems},
  pages     = {1119--1127},
  year      = {2024},
  publisher = {IFAAMAS},
  url       = {https://www.ifaamas.org/Proceedings/aamas2024/pdfs/p1119.pdf}
}

@article{rahwan2009anytime,
  author  = {Rahwan, Talal and Ramchurn, Sarvapali D. and Jennings, Nicholas R. and Giovannucci, Andrea},
  title   = {An Anytime Algorithm for Optimal Coalition Structure Generation},
  journal = {Journal of Artificial Intelligence Research},
  volume  = {34},
  pages   = {521--567},
  year    = {2009},
  doi     = {10.1613/jair.2690}
}

@article{holm1979multiple,
  author  = {Holm, Sture},
  title   = {A Simple Sequentially Rejective Multiple Test Procedure},
  journal = {Scandinavian Journal of Statistics},
  volume  = {6},
  number  = {2},
  pages   = {65--70},
  year    = {1979},
  doi     = {10.2307/4615733},
  url     = {https://www.jstor.org/stable/4615733}
}

@book{efron1993bootstrap,
  author    = {Efron, Bradley and Tibshirani, Robert J.},
  title     = {An Introduction to the Bootstrap},
  publisher = {Chapman and Hall/CRC},
  year      = {1993},
  doi       = {10.1201/9780429246593}
}

@article{spearman1904association,
  author  = {Spearman, Charles},
  title   = {The Proof and Measurement of Association between Two Things},
  journal = {The American Journal of Psychology},
  volume  = {15},
  number  = {1},
  pages   = {72--101},
  year    = {1904},
  doi     = {10.2307/1412159}
}

@article{lakens2013effects,
  author  = {Lakens, Dani{\"e}l},
  title   = {Calculating and Reporting Effect Sizes to Facilitate Cumulative Science: A Practical Primer for $t$-Tests and {ANOVA}s},
  journal = {Frontiers in Psychology},
  volume  = {4},
  pages   = {863},
  year    = {2013},
  doi     = {10.3389/fpsyg.2013.00863},
  url     = {https://www.frontiersin.org/journals/psychology/articles/10.3389/fpsyg.2013.00863/full}
}

@article{harris2020numpy,
  author  = {Harris, Charles R. and Millman, K. Jarrod and van der Walt, St{\'e}fan J. and Gommers, Ralf and Virtanen, Pauli and Cournapeau, David and Wieser, Eric and Taylor, Julian and Berg, Sebastian and Smith, Nathaniel J. and Kern, Robert and Picus, Matti and Hoyer, Stephan and van Kerkwijk, Marten H. and Brett, Matthew and Haldane, Allan and Fern{\'a}ndez del R{\'i}o, Jaime and Wiebe, Mark and Peterson, Pearu and G{\'e}rard-Marchant, Pierre and Sheppard, Kevin and Reddy, Tyler and Weckesser, Warren and Abbasi, Hameer and Gohlke, Christoph and Oliphant, Travis E.},
  title   = {Array Programming with {NumPy}},
  journal = {Nature},
  volume  = {585},
  pages   = {357--362},
  year    = {2020},
  doi     = {10.1038/s41586-020-2649-2}
}

@article{virtanen2020scipy,
  author  = {Virtanen, Pauli and Gommers, Ralf and Oliphant, Travis E. and Haberland, Matt and Reddy, Tyler and Cournapeau, David and Burovski, Evgeni and Peterson, Pearu and Weckesser, Warren and Bright, Jonathan and van der Walt, St{\'e}fan J. and Brett, Matthew and Wilson, Joshua and Millman, K. Jarrod and Mayorov, Nikolay and Nelson, Andrew R. J. and Jones, Eric and Kern, Robert and Larson, Eric and Carey, C. J. and Polat, {\.I}lhan and Feng, Yu and Moore, Eric W. and VanderPlas, Jake and Laxalde, Denis and Perktold, Josef and Cimrman, Robert and Henriksen, Ian and Quintero, E. A. and Harris, Charles R. and Archibald, Anne M. and Ribeiro, Ant{\^o}nio H. and Pedregosa, Fabian and van Mulbregt, Paul and {{SciPy} 1.0 Contributors}},
  title   = {{SciPy} 1.0: Fundamental Algorithms for Scientific Computing in {Python}},
  journal = {Nature Methods},
  volume  = {17},
  pages   = {261--272},
  year    = {2020},
  doi     = {10.1038/s41592-019-0686-2}
}

@inproceedings{hagberg2008networkx,
  author    = {Hagberg, Aric A. and Schult, Daniel A. and Swart, Pieter J.},
  title     = {Exploring Network Structure, Dynamics, and Function Using {NetworkX}},
  booktitle = {Proceedings of the 7th Python in Science Conference},
  pages     = {11--15},
  year      = {2008},
  doi       = {10.25080/TCWV9851}
}

@article{waskom2021seaborn,
  author  = {Waskom, Michael L.},
  title   = {seaborn: Statistical Data Visualization},
  journal = {Journal of Open Source Software},
  volume  = {6},
  number  = {60},
  pages   = {3021},
  year    = {2021},
  doi     = {10.21105/joss.03021}
}

@inproceedings{das2019tarmac,
  author    = {Das, Abhishek and Gervet, Th{\'e}ophile and Romoff, Joshua and Batra, Dhruv and Parikh, Devi and Rabbat, Mike and Pineau, Joelle},
  title     = {{TarMAC}: Targeted Multi-Agent Communication},
  booktitle = {Proceedings of the 36th International Conference on Machine Learning},
  series    = {Proceedings of Machine Learning Research},
  volume    = {97},
  pages     = {1538--1546},
  publisher = {PMLR},
  year      = {2019},
  url       = {https://proceedings.mlr.press/v97/das19a.html}
}
